\documentclass[12pt,a4paper]{amsart}
\usepackage[style=numeric, 
            giveninits=true, 
            sortcites=true, 
            maxalphanames=4, 
            maxbibnames=99, 
            maxcitenames=99
           ]{biblatex}
\usepackage{amsmath}
\usepackage{amscd}
\usepackage{latexsym}
\usepackage{amsfonts}
\usepackage{amssymb}
\usepackage{amsthm}
\usepackage{graphicx}
\usepackage{verbatim}
\usepackage{extarrows}
\usepackage{mathrsfs}
\usepackage{enumerate}
\usepackage{hyperref}
\usepackage{enumerate}
\usepackage{color}

\usepackage{tensor}

\usepackage{biblatex} 
\usepackage[T1]{fontenc}
\usepackage{lmodern} 

\DeclareNameAlias{default}{family-given}
\DeclareDelimFormat{multinamedelim}{\addcomma\space} 
\DeclareDelimFormat{finalnamedelim}{\addcomma\space} 
 
\renewbibmacro{in:}{}

\DeclareFieldFormat[article,inproceedings,incollection,misc,unpublished]{title}{#1} 

\DeclareFieldFormat{journaltitle}{\mkbibemph{#1}}
\DeclareFieldFormat{booktitle}{\mkbibemph{#1}}

\DeclareFieldFormat{volume}{\mkbibbold{#1}}
\DeclareFieldFormat{pages}{#1}
\DeclareFieldFormat{year}{(#1)}

\DeclareBibliographyDriver{article}{%
  \printnames{author}
  \setunit{\addperiod\space}
  \newblock
  \printfield{title}
  \setunit{\addperiod\space}
  \newblock
  \printfield{journaltitle}
  \setunit{\addspace}
  \printtext{\mkbibbold{\printfield{volume}}}
  \printfield{number}
  \setunit{\addcomma\space}
  \printfield{pages}
  \printfield{year}
  \setunit{\addperiod\space}
  \printfield{doi}
  \finentry
}

\newcommand{\nn}{\mathfrak{n}}

\newcommand{\htimes}{\,\hat \otimes\,}

\newcommand{\tn}{\tilde \nn}
\newcommand{\tA}{\tilde A}
\newcommand{\te}{\tilde e}
\newcommand{\tU}{\tilde U}
\newcommand{\tH}{\tilde \HH}
\newcommand{\tSigma}{\tilde \Sigma} 
\newcommand{\tV}{\tilde V}

\newcommand{\Si}[2]{\Sigma\indices{_{#1}^{#2}}}

\newcommand{\cg}{\mathring g}
\newcommand{\ch}{\mathring h}
\newcommand{\ctau}{\mathring \tau}

\newcommand{\cee}{\mathring e} 
\newcommand{\cnn}{\mathring \nn}
\newcommand{\cA}{\mathring A}
\newcommand{\cU}{\mathring U}
\newcommand{\cHH}{\mathring \HH}
\newcommand{\cSigma}{\mathring \Sigma}

\newcommand{\XXX}{\mathfrak X}
\newcommand{\AAA}{\mathcal A}
\newcommand{\BBB}{\mathcal B}

\newcommand{\BB}{\mathbf B}
\newcommand{\tBB}{\tilde \BB}

\newcommand{\scrB}{\mathscr B}
\newcommand{\WW}{{\mathbf W}}
\newcommand{\tWW}{\tilde {\mathbf W}}
\newcommand{\tFF}{\tilde {\mathbf F}}

\newcommand{\gini}{\mathfrak g} 
\newcommand{\kini}{\mathfrak k}

\newcommand{\angbracket}[1]{\left\langle #1 \right\rangle}

\usepackage{Commands_Dong}

\title[Stability of Kasner singularities in $(3+1)$-Dim vacuum universe]{The past stability of Kasner singularities for the $(3+1)$-dimensional Einstein vacuum spacetime under polarized $U(1)$-symmetry}
\author{Kai Dong}
\address{Department of Mathematics, National University of Singapore}
\email{kdong@u.nus.edu}

\begin{document}

\begin{abstract}
    In this paper, we give a new proof to \revise{a past stability result established in Fournodavlos-Rodnianski-Speck \cite{fournodavlos2023stable}, for Kasner solutions of the $(3+1)$-dimensional Einstein vacuum equations under polarized $U(1)$-symmetry}. \revise{Our method, inspired by Beyer-Oliynyk-Olvera-Santamar{\'\i}a-Zheng \cite{beyer2021fuchsian, beyer2025localized}, relies on a newly developed $(2+1)$ orthonormal-frame decomposition and a careful symmetrization argument, after which the Fuchsian techniques can be applied.} 
    
    We show that the perturbed solutions are \revise{asymptotically pointwise Kasner}, geodesically incomplete and crushing at the Big Bang singularity. 
    They are achieved by reducing the $(3+1)$ Einstein vacuum equations to a Fuchsian system coupled with several constraint equations, with the symmetry assumption playing an important role in the reduction. 
    Using Fuchsian theory together with finite speed of constraints propagation, we obtain global existence and precise asymptotics of the solutions up to the singularities. 
\end{abstract}

\maketitle

\section{Introduction}

The Big Bang solutions have been studied extensively in both mathematics and physics for decades. 
These spacetimes furnish standard cosmological models that describe large-scale expansion and singular behavior of the universe. Among them, the Kasner solutions form a simple, explicit class exhibiting Big Bang singularities; they therefore provide a natural testing ground for questions of stability and asymptotic behavior. 

Fix the spatial topology to be the closed manifold $\mathbb T^3$. The goal of this paper is to revisit the polarized $U(1)$-symmetric Kasner stability problem for Einstein vacuum equations treated in Fournodavlos-Rodnianski-Speck \cite{fournodavlos2023stable}. \revise{We will present an alternative approach} that emphasizes a $(2+1)$ orthonormal-frame reduction and a tailored Fuchsian formulation. 
Our analysis is informed by the recent Fuchsian techniques developed in Beyer-Oliynyk-Olvera-Santamar{\'\i}a-Zheng \cite{beyer2021fuchsian, beyer2025localized}. 
Concretely, we consider the Cauchy problem with initial data on a $t=t_0$ slice obtained by symmetric perturbations of Kasner data.  
Evolving them by the Einstein vacuum equations, we prove that the resulting spacetimes develop Big Bang singularities at the synchronous time $t=0$ while remaining asymptotically pointwise Kasner near these singularities.
\revise{We also remark an earlier work by Alexakis-Fournodavlos \cite{alexakis2025stable}, where the authors established the stability of Schwarzschild-type singularities for axi-symmetric and polarized Einstein vacuum spacetimes with spatial topology $\mathbf R^3$, via a PDE-ODE approach, energy estimates and an iteration scheme.}

\subsection{Einstein field equations and Kasner spacetimes}

On a $(D+1)$-dimensional Lorentzian manifold the Einstein field equations read 
\begin{subequations}
    \begin{gather}
        ^{(D+1)}R_{ab} - \frac 12 {^{(D+1)}R}g_{ab} = T_{ab}, \\ 
        ^{(D+1)}\nabla^a T_{ab} = 0, 
    \end{gather}
\end{subequations}
where $^{(D+1)}R_{ab}$ and $^{(D+1)}R$ denote the Ricci and scalar curvatures, and $T_{ab}$ is the stress-energy tensor. 

In this work, we treat the Einstein vacuum equations, with $$T_{ab} = 0,$$ 
as well as the Einstein-scalar field equations, where 
\begin{gather*}
    T_{ab} = 2 \nabla_a \phi \nabla_b \phi - g_{ab} \nabla^c \phi \nabla_c \phi, \\ 
    \Box_g \phi = 0, 
\end{gather*}
combined with the conservation law of scalar field $\phi$. 

In cosmology, the Kasner spacetime provides explicit solutions in both settings. The $(D+1)$-dimensional Kasner solution for the Einstein-scalar field system can be written as 
\begin{equation} \label{g_Kas} 
    g_{Kas} = - (\mathrm dt)^2 + \sum_{\Lambda=1}^D t^{2p_\Lambda} (\mathrm dx^\Lambda)^2, \quad \phi = B \log t, 
\end{equation}
where the exponents $p_\Lambda$ are called the \emph{Kasner exponents}, satisfying two algebraic relations 
\begin{gather}
    \sum_{\Lambda=1}^D p_\Lambda = 1, \quad \sum_{\Lambda=1}^D p_\Lambda^2 = 1 - B^2. \label{Kasner_Ddim}
\end{gather}
The special case $B=0$ then yields the vacuum Kasner metrics. 

In particular, for $D=3$ and $\phi \equiv 0$, the Kasner relations (\ref{Kasner_Ddim}) reduce to 
\begin{equation} \label{p123} 
    \sum_{I=1}^3 p_I = \sum_{I=1}^3 p_I^2 = 1. 
\end{equation}
A straightforward algebraic check shows that exactly one of $p_i$'s is negative, and we adopt the ordering 
\begin{equation} 
    p_1 < 0 < p_2 \, \& \, p_3.
\end{equation}

As a guiding example, near its spacelike singularity, the well-known Schwarzschild solution is locally approximated by the Kasner metric with exponents 
\begin{equation}
    (p_1,p_2,p_3) = \left(-\frac 13, \, \frac 23, \, \frac 23 \right), \label{Schwarzschild}
\end{equation}
where the positive exponents correspond to angular directions. 
This motivates, among other things, restricting attention to polarized $U(1)$-symmetry in which the symmetry direction is associated with a positive Kasner exponent. 

    The recent work by Fournodavlos-Rodnianski-Speck \cite{fournodavlos2023stable} studied the stability of Kasner metrics for Einstein vacuum (dimension $D \ge 10$) and Einstein-scalar field equations (dimension $D \ge 3$) in the full \emph{subcritical regime}: 
    \begin{equation}
        \max_{I \ne J} \{p_I + p_J - p_K \} < 1, \label{subcritcal} 
    \end{equation}
    and the $(3+1)$ Einstein-scalar field equations were also revisited in Beyer-Oliynyk-Zheng \cite{beyer2025localized} with a localized argument via Fuchsian techniques. 
    For the Einstein-Maxwell-scalar field-Vlasov system, An-He-Shen \cite{an2025stabilitybigbangsingularity} further generalized to the stability of Kasner singularities in the full \emph{strongly subcritical regime}: 
    \begin{equation} \label{StronglySubcritical} 
        \max \{p_I + p_J - p_K \} < 1. 
    \end{equation}
    However, the $(3+1)$-dimensional vacuum Kasner exponents satisfying (\ref{p123}) are \revise{\emph{neither} strongly subcritical \emph{nor} subcritical}, so the analysis here must confront this non-subcritical situation.

\subsection{The spatial topology and symmetry assumption of spacetimes} 

As introduced in \cite{isenberg2002asymptotic, fournodavlos2023stable} and also presented in \defref{Symmetry}, we assume the spacetime admits a nowhere degenerating Killing vector field $S$ with $\mathbb T^1$ orbits, that is hypersurface-orthogonal, spacelike and satisfies 
$$\mathcal L_S g = \mathcal L_S k = 0.$$
Moreover, $k(S,X) = 0$ for every $\Sigma_t$-tangent vector field $X$ with $g(S,X) = 0$, where $k$ denotes the second fundamental form of the spatial slices $\Sigma_t \simeq \mathbb T^3$. 

Under this symmetry, the metric splits as 
\begin{equation} \label{metric_gh} 
    g = h + e^{2\phi} (\mathrm d x^3)^2, \quad h = - \nn^2 (\mathrm dt)^2 + \sum_{I,J=1}^2 g_{IJ} \mathrm d x^I \mathrm d x^J 
\end{equation}
with all components independent of $x^3$. The $x^3$-direction corresponds to the positive Kasner exponent $p_3$.

\subsection{Rough version of the main result}
Below is an informal description of the main theorem; a precise statement appears in Section~\ref{Section_the_main_theorem} \thmref{Main}. 

\begin{theorem}[Rough version]
    Consider the $(3+1)$-dimensional Einstein Vacuum Equations under polarized $U(1)$-symmetry. Then every nontrivial vacuum Kasner solution (\ref{g_Kas}) with $B=0$ and $D=3$ is dynamically stable under small polarized $U(1)$-symmetric perturbations near their Big Bang singularity. 

    More specifically, for sufficiently small perturbations whose size is controlled by $\delta_0>0$, the perturbed spacetimes developing synchronized singularities at the harmonic time $t=0$, are asymptotically pointwise Kasner, and exhibit crushing behavior with $\tr k \sim t^{-1 + c\delta_0}$, and are past timelike geodesically incomplete with Kretschmann scalar blowing up like $R^{abcd} R_{abcd} \sim t^{-4 + c\delta_0}$, and are asymptotically velocity-terms dominated (AVTD) as $t \rightarrow 0^+$. 
\end{theorem}

\subsection{Prior and Related results}

The mathematical study of cosmological singularities and Big Bang asymptotics has a long history, with contributions from both mathematical relativity and theoretical physics. At a heuristic level, the BKL picture \cite{belinskii1970oscillatory} initiated this line of inquiry; the BKL heuristics propose that near a generic spacelike singularity, the dynamics are approximately local and governed, at leading order, by ODEs in time. This physical viewpoint motivated a series of rigorous investigations for different models. 

On the mathematical side, important early progress concerned spacetimes with Gowdy symmetry, see e.g. \cite{chrusciel1990strong,isenberg1990asymptotic}; influential results on quiescent singularities and AVTD behavior were obtained in \cite{andersson2001quiescent}. There are also active literature on nonlinear stability and singularity formation for specific cosmological models; representative contributions include a sequence of works by Rodnianski-Speck \cite{rodnianski2009stability, rodnianski2013nonlinear, rodnianski2014stablebigbangformation, rodnianski2018regime, rodnianski2009stability}, major advances by Klainerman-Szeftel \cite{klainerman2017global}, several comprehensive frameworks developed by Ringstr\"{o}m \cite{ringstrom2009strong, ringstrom2010cosmic, ringstrom2013topology, ringstrom2022initial, ringstrom2025initial}. 
\revise{See also related analyses \cite{taylor2017global, lindblad2020global, bigorgne2021asymptotic, fajman2021stabilizing, fajman2021stability, fajman2024stability, fajman2025cosmic, fajman2025stability, urban2024quiescentbigbangformation}. 
Regarding the study of spacelike singularities for the Einstein-scalar field system with spherical symmetry, we refer to \cite{an2020polynomial, an2023curvature, an2025quantitative, li2025kasner}.} 

The Fuchsian method in PDE theory, originally developed in a semi-linear, analytic setting, has classical roots, see e.g. \cite{bove1985cauchy, kichenassamy1993blow1, kichenassamy1993blow2, kichenassamy1996fuchsian, kichenassamy2007fuchsian}. The geometric adaptation of these ideas to Einstein equations were initiated by Rendall and St{\aa}hl in \cite{rendall2000fuchsian, rendall2004asymptotics, staahl2002fuchsian}. Building on these foundations, the analytical assumptions extended to quasi-linear and Sobolev-regular settings, see works by Ames-Beyer-Isenberg-LeFloch \cite{ames2012quasi, ames2013quasilinear, beyer2010second, beyer2011second, beyer2017self}.

On the physics side, the orthonormal-frame decomposition in $(3+1)$ cosmology goes back to Ellis \cite{ellis2012relativistic}, and further discussions in \cite{van1997general, rohr2005conformal}. 
This decomposition is convenient for writing Einstein field equations in the first-order form and for separating time-dominant from spatial terms near singularities. 
The study of Einstein equations with spacelike $U(1)$ isometry groups has its origins in the physics literature as well, see e.g. \cite{moncrief1986reduction, isenberg2002asymptotic}.

More directly related to the present paper are recent \revise{important works} proving stability of various cosmological backgrounds, notably: 
\begin{itemize}
    \item \cite{fournodavlos2023stable}: stability of Kasner-like backgrounds in the full subcritical regime (\ref{subcritcal}) and for polarized $U(1)$-symmetry; the polarized $U(1)$-symmetry is the setting we revisit here; \cite{an2025stabilitybigbangsingularity}: generalization of \cite{fournodavlos2023stable} to the Einstein-Maxwell-scalar field-Vlasov system in the full strongly subcritical regime (\ref{StronglySubcritical}); 
    \item Major works by Beyer-Oliynyk-Olvera-Santamar{\'\i}a \cite{beyer2021fuchsian, oliynyk2021fuchsian}, and then a sequence of recent works apply the Fuchsian techniques to past stability problems, including \cite{beyer2020relativistic, beyer2021fuchsian, beyer2023past, fajman2021stabilizing, fajman2024stability, oliynyk2021fuchsian, Todd2024, beyer2025localized, fournodavlos2025future}, and future stability problems, including  \cite{oliynyk2021future, liu2024future, fournodavlos2025future}. 
\end{itemize}

\subsection{Proof strategy and New ingredients}

Our method is inspired mainly by \cite{beyer2021fuchsian, beyer2025localized}, and further combines a newly developed $(2+1)$ orthonormal-frame reduction and a careful symmetrization procedure. 
The proof strategy proceeds through the following steps: 

\begin{itemize}

    \item[1.] {\em Reduction and frame decomposition.} In Sections \ref{Section_Formulation} and \ref{Section_Fuchsian_system}, we reduce the $x^3$ symmetry direction, and work on the quotient $(M^{3+1}/\mathbb T^1,h)$ with $h$ as in (\ref{metric_gh}). 
    
    Choosing a Fermi-Walker transported orthonormal frame $\{e_0, e_1, e_2\}$ with $e_0 = \nn^{-1} \pa_t$, as a generalization of \cite{ellis2012relativistic, van1997general}, we employ \revise{the newly developed (irrotational) relativistic $(2+1)$ orthonormal frame decomposition}: 
    \begin{subequations}
        \begin{align}
            [e_0, e_A] &= U_A e_0 - \left( \HH \delta_A^B + \Si AB \right) e_B, \label{intro_e0eA} \\ 
            [e_A, e_B] &= 2 A_{[A} \delta_{B]}^C e_C, \label{intro_eAeB} 
    \end{align}
    \end{subequations}
    which reduces the unknowns to a vector $\WW := \left( \nn, e_P^\Lambda, U_P, A_P, \HH, \Sigma_{PQ} \right)^T$. 
    \revise{In our setting,} the Einstein field equations give an evolution system for $\WW$ of the form: 
    \begin{align} \label{intro_WWequation}
        \BB^0 \pa_t \WW + \nn e_C^\Lambda \BB^C \pa_\Lambda \WW = \GG (t, \WW), 
    \end{align}
    (see also (\ref{evolution_nn})-(\ref{evolution_Sigma}) or (\ref{HyperbolicSystemWW})), combined with constraint equations 
    \begin{equation} \label{intro_XXXdefinition}
        \XXX = \left( \EEE_{AB}^\Lambda, \UUU_{AB}, \NNN_{A}, \MMM_{A}, \HHH \right)^T = 0, 
    \end{equation}
    (see (\ref{constraint_e})-(\ref{constraint_A}) or (\ref{def_constraintEEE})-(\ref{def_constraintHHH}) for the explicit formulas of $\XXX$). \revise{We remark that the above decomposition (\ref{intro_e0eA}), (\ref{intro_eAeB}) and the full associated computations in Section \ref{Section_Formulation} are adapted to both general polarized $U(1)$-symmetric spacetimes and $(2+1)$ Lorentzian manifolds as well}; 
    
    \item[2.] {\em Propagation of constraints.} In Section \ref{Section_Propagation_of_Constraints}, we show that if the constraints (\ref{intro_XXXdefinition}) hold initially, then they are preserved for as long as the solution $\WW$ exists. 
    
    \revise{A key observation is that, within the $(2+1)$ decomposition, the symmetry of the Ricci tensor $^{(2+1)}R_{AB}$ is not immediate and must be recovered using carefully the Jacobi identity. After performing the necessary symmetrization, 
    the constraints satisfy a linear strongly hyperbolic system of the form: 
    \begin{align} \label{intro_XXXequation}
        \AAA^0 \pa_t \XXX + \nn \AAA^R e_R^\Lambda \pa_\Lambda \XXX = \BBB \XXX 
    \end{align}
    (see also (\ref{propagation_E})-(\ref{propagation_H}) or (\ref{XXXequation})) and the finite speed of propagation then implies $\XXX \equiv 0$. On the other hand, without the symmetrization step, the propagation system (\ref{XXXequation}) would \emph{fail} to be strongly hyperbolic or Friedrichs-symmetrizable, creating significant difficulties and obstacles, see {\emph {Remark}~\ref{KeySymmetrization}} for details;} 
    
    \item[3.] {\em Symmetrization and Fuchsian method.} In Sections \ref{Section_Fuchsian_system} and \ref{Section_Fuchsian_Theorem}, we perform algebraic symmetrizations and time-weighted rescalings to transform the evolution system (\ref{intro_WWequation}) into a \emph{Fuchsian form}: 
    \begin{equation} 
        \tBB^0 \pa_t \tWW + \tBB^\Lambda \pa_\Lambda \tWW = \frac 1t \scrB \PP \tWW + \frac 1t \tFF[t, \tWW], 
    \end{equation}
    where $\tWW$ is the rescaled unknown, $\PP$ is a projection used to identify the \emph{decaying} and \emph{converging} components of $\tWW$, and other coefficient matrices and nonlinear source terms satisfy the \emph{Fuchsian conditions} (see items (1) to (5) in the end of statement of \thmref{FuchsianTheorem} or the proof to \propref{Prop_Fuchsian} for the precise Fuchsian conditions). This reduces the study of global existence and asymptotics for singular initial value problems to verification of several algebraic and analytic conditions. 
    Following prior works \cite{beyer2025localized, beyer2021fuchsian}, the simplified Fuchsian method that is more tailored to our setting is presented in Appendix \ref{Appendix_Fuchsian}; 
    
    \item[4.] {\em Proof of the main theorem.} Finally, in Section \ref{Section_the_main_theorem}, the Fuchsian existence and decay yield precise asymptotic expansions for the components of the rescaled unknown $\tWW$: 
    \begin{equation}
        \begin{gathered}
            \tn = t^{2\tH_*+\eps_0} \tn_* \left( 1 + O_{H^{s-1}}(t^\zeta) \right), \\ 
            \te_I^\Lambda = O_{H^{s-1}}(t^\zeta), \quad \tA_P = O_{H^{s-1}}(t^\zeta), \quad \tU_P = O_{H^{s-1}}(t^\zeta), \\ 
            \tH = \tH_* + O_{H^{s-1}}(t^\zeta), \quad \tSigma_{AB} = \tSigma_{*AB} + O_{H^{s-1}}(t^\zeta), 
        \end{gathered}
    \end{equation}
    where the quantities with index $*$ are independent of $t$. From these profiles, we derive the asymptotics of the second fundamental form and curvature invariants associated with the physical metric $g$, and deduce AVTD behavior from the leading order ODE dynamics; see \thmref{Main}. 
\end{itemize}

\subsection{Acknowledgments}

The author would like to thank Prof.\ Xinliang An and Prof.\ Todd A. Oliynyk for many illuminating discussions. The author is also grateful to Dr.\ Haoyang Chen for his careful reading and valuable comments on the manuscript. 
\smallskip


\section{Formulation of the equations} \label{Section_Formulation}

In this section, we present some preliminaries and the formulations, including the main reduction to the $(2+1)$ orthonormal-frame hyperbolic system. 

Consider the $(3+1)$-decomposition of the sliced spacetime $(M^{3+1}, g)$ with zero shift. We write its metric as 
\begin{equation} \label{metric0}
    g = -\nn^2 \dd t^2 + g_{IJ} \dd x^I \mathrm d x^J,
\end{equation}
where the function $\nn>0$ is called the \emph{lapse} function. 

Let $^{(3+1)}\nabla$ denote the Levi-Civita connection of $g$, and let $^{(3+1)}R$ denote the associated curvature tensors. Then the $(3+1)$-dimensional Einstein vacuum equation reads 
\begin{equation}\label{EVEoriginal}
    {^{(3+1)}R}_{ij}(g) = 0. 
\end{equation} 

We employ the following polarized $U(1)$-symmetry assumption and the wave map reduction of the physical metric $g$ throughout this paper.

\begin{definition} \label{Symmetry} 
    The spacetime $(M^{3+1}, g)$ is called polarized $U(1)$-symmetric with respect to the $x^3$-direction, if the vector field $\pa_3$ generates a spatial hypersurface-orthogonal Killing symmetry, and the metric is independent of $x^3$. Then, with a slight abuse of notation, the metric can be written in the form: 
    \begin{equation} \label{metric}
    g = - \nn^2 \dd t^2 + \sum_{I,J = 1}^2 g_{IJ} \dd x^I \dd x^J + e^{2\phi} (\mathrm d x^3)^2 =: h + e^{2\phi} (\mathrm d x^3)^2, 
    \end{equation}
    where, $\phi$ is the wave map from $M^{3+1}$ to $M^{3+1}/\mathbb T^1$, $h$ denotes a $(2+1)$-dimensional Lorentzian metric on the quotient manifold $M^{3+1}/\mathbb T^1$ and all components of $h$, including $\phi$, are functions on $M^{3+1}/\mathbb T^1$, independent of $x^3$. 
\end{definition}

The following proposition is the core reduction in our analysis. 

\begin{proposition}
    Let $(M^{3+1}, g)$ be a Lorentzian manifold with polarized $U(1)$-symmetry. 
    Then the $(3+1)$-dimensional Einstein vacuum equation for $g$: 
    \begin{equation}
        {^{(3+1)}R}_{ab} (g) = 0, 
    \end{equation}
    is equivalent to the $(2+1)$-dimensional modified Einstein-scalar-field equations for $h$ and $\phi$: 
    \begin{subequations}
        \begin{align}
            {^{(2+1)}R}_{ab}(h) &= \, ^{(2+1)}\nabla_a \, ^{(2+1)}\nabla_b \phi + \, ^{(2+1)}\nabla_a \phi \, ^{(2+1)}\nabla_b \phi, \label{ESForiginal_1} \\ 
            \Box_h \phi &= - \norm{^{(2+1)}\nabla \phi}_h^2, \label{ESForiginal_2}
        \end{align}
    \end{subequations}
    where $^{(2+1)}\nabla$ and $^{(2+1)}R$ denote the Levi-Civita connection and curvature tensors of $h$ respectively. 
    Moreover, in what follows, unless otherwise specified, we drop the index $(2+1)$ on all geometric quantities of $M^{3+1}/\mathbb T^1$, for instance, $^{(2+1)}R_{abcd}$ will simply be denoted by $R_{abcd}$. 
\end{proposition}

\begin{proof}
    The proof is a direct computation. We have the following relations between the Christoffel symbols $^{(3+1)}\Gamma$ and $\Gamma$ associated with $g$ and $h$ respectively: 
    \begin{equation}
        \begin{gathered}
            ^{(3+1)}\Gamma^i_{jk} = \Gamma^i_{jk}, \quad ^{(3+1)}\Gamma^3_{33} = 0, \quad ^{(3+1)}\Gamma^3_{jk} = 0, \quad ^{(3+1)}\Gamma^i_{3k} = 0, \\ 
            ^{(3+1)}\Gamma^3_{j3} = \nabla_j \phi, \quad ^{(3+1)}\Gamma^i_{33} = - e^{2\phi} \nabla^i \phi; 
        \end{gathered}
    \end{equation}
    where the indices $i,j,k$ (and all subsequent indices) take values in $\{0,1,2\}$. Hence the Riemannian and Ricci tensor components $^{(3+1)}R$ of $g$ in the orthonormal frame are given by 
    \begin{equation}
        \begin{gathered}
            ^{(3+1)}R_{ijkl} = R_{ijkl}, \quad 
            ^{(3+1)}R_{ijk3} \equiv 0, \quad 
            ^{(3+1)}R_{3jk3} = \nabla^2_{j,k} \phi + \nabla_j \nabla_k \phi, \\ 
            ^{(3+1)}R_{ik} = R_{ik} - \nabla^2_{i,k} \phi - \nabla_i \phi \nabla_k \phi, \quad 
            ^{(3+1)}R_{3k} \equiv 0, \quad 
            ^{(3+1)}R_{33} = -e^{2\phi} \, \left( \Box_{h} \phi + |\nabla \phi|_{h}^2 \right). 
        \end{gathered}
    \end{equation}
    The Einstein equations then yield (\ref{ESForiginal_1}) and (\ref{ESForiginal_2}). 
\end{proof}

From now on, unless otherwise specified, we will perform all subsequent calculations on the quotient spacetime $(M^{3+1}/\mathbb T^1, h)$. 

\subsection{The relativistic $(2+1)$-dimensional orthonormal frame decomposition}

Let $\{e_I = e_I^\Lambda \pa_\Lambda \}_{I=1,2}$ be a spatial orthonormal frame for $h$ on the initial hypersurface $\Sigma_{t_0} := \{t = t_0\}$, and let $e_0 = \nn^{-1} \pa_t$ be the unit timelike vector field defined for all $t$ and normal to $\Sigma_{t}$. We propagate the spatial vectors along $e_0$ using Fermi-Walker transport: 
\begin{definition}
    A vector field $X^i$ is called Fermi-Walker transported along a timelike curve $\gamma(s)$ if its Fermi-Walker derivative vanishes, i.e. 
    \begin{equation}
        \frac {\mathrm D_F X^i}{\dd s} := Z^j \nabla_j X^i + 2 A^{[i} Z^{j]} X_j = 0, 
    \end{equation}
    where $Z^i=(\pa_s)^i$, $A^i=Z^j\nabla_j Z^i$. 
\end{definition}

By transporting the frame $\{e_1, e_2\}$ along the $t$-curves with tangent vector $Z^i = e_0 \delta^{0i}$ and imposing Fermi-Walker transport, we have that 
$$0 = \nabla_{e_0} e_I + h(e_0,e_I) \nabla_{e_0} e_0 - h(e_I, \nabla_{e_0} e_0) e_0.$$
Since $h(e_0,e_I) = 0$, it follows that 
\begin{equation} \label{FWtrans}
    \nabla_{e_0} e_I = h(\nabla_{e_0} e_0, e_I) e_0. 
\end{equation}
Hence we obtain a Fermi-Walker transported orthonormal frame $\{e_0, e_1, e_2\}$ on the spacetime. In this frame, the metric components satisfy 
    $$h_{ij} = h(e_i, e_j) = \eta_{ij},$$ 
where $\eta$ denotes the standard Minkowski metric. 
The principal advantage of Fermi-Walker transport is to eliminate arbitrary spins or rotations of the frame along the flow of $e_0$, which simplifies the subsequent decompositions; 
see \cite{ellis2012relativistic, van1997general} for $(3+1)$ treatments for tetrads with self-rotation.

With the spin freedom eliminated, the relativistic $(2+1)$ orthonormal-frame decomposition proceeds as follows: 
\begin{subequations}
    \begin{align}
        [e_0, e_A] &= U_A e_0 - \left( \HH \delta_A^B + \Si AB \right) e_B, \label{e0eA} \\ 
        [e_A, e_B] &= 2 A_{[A} \delta_{B]}^C e_C, \label{eAeB} 
    \end{align}
\end{subequations}
where $\HH$ is a function, $A$ and $U$ are $1$-tensors and $\Sigma$ is a symmetric, trace-free $2$-tensor. 

Let $\gamma$ denote the connection coefficients 
$$\gamma_{acb} := g(\nabla_{e_a} e_b, e_c),$$ 
then from (\ref{e0eA}) and (\ref{eAeB}), we compute explicitly: 
\begin{subequations}
    \begin{gather}
        \gamma_{00A} = - \gamma_{0A0} = -U_A, \\ 
        \gamma_{A0B} = - \gamma_{AB0} = - \left( \HH \delta_{AB} + \Sigma_{AB} \right), \\ 
        \gamma_{ABC} = 2A_{[B} \delta_{C]A}, \\ 
        \gamma_{000} = 0, \quad \gamma_{A00} = 0, \quad \gamma_{0AB} = 0. 
    \end{gather}
\end{subequations}

Since the curvature components are given by 
\begin{equation}
    R_{abcd} = - e_a (\gamma_{bdc}) + e_b (\gamma_{adc}) + \eta^{ij} \left( \gamma_{ajc} \gamma_{bdi} - \gamma_{bjc} \gamma_{adi} + \gamma_{ajb} \gamma_{idc} - \gamma_{bja} \gamma_{idc} \right), 
\end{equation}
a direct computation yields: 

\begin{lemma} \label{Ricci}
    The $(2+1)$-dimensional metric $h$ has the following nontrivial Riemannian and Ricci curvature components: 
    \begin{subequations}
        \begin{align}
            R_{0B0D} &= -e_0(\HH) \delta_{BD} - e_0(\Sigma_{BD}) + e_B(U_D) + U_B U_D + A_D U_B - A \cdot U \delta_{BD} \\ 
            &\quad - \HH^2 \delta_{BD} - 2 \HH \Sigma_{BD} - \Si BC \Sigma_{CD} \nonumber, \\ 
            R_{0BCD} &= e_0(A_C) \delta_{BD} - e_0(A_D) \delta_{BC} + \HH U_C \delta_{BD} + U_C \Sigma_{BD} - \HH U_D \delta_{BC} - U_D \Sigma_{BC} \\ 
            &\quad - \HH A_D \delta_{BC} + \HH A_C \delta_{BD} - \Sigma_{BC} A_D + \Sigma_{BD} A_C, \nonumber \\ 
            R_{ABCD} &= - e_A(A_D) \delta_{CB} + e_A(A_C) \delta_{DB} + e_B(A_D) \delta_{CA} - e_B(A_C) \delta_{DA} \\ 
            &\quad + (\HH^2 - A \cdot A) (\delta_{AC} \delta_{BD} - \delta_{BC} \delta_{AD}) \nonumber \\ 
            &\quad + \HH (\delta_{AC} \Sigma_{BD} + \delta_{BD} \Sigma_{AC} - \delta_{BC} \Sigma_{AD} - \delta_{AD} \Sigma_{BC}) + \Sigma_{AC} \Sigma_{BD} - \Sigma_{BC} \Sigma_{AD} \nonumber \\ 
            R_{00} &= -2 e_0(H) + e_A(U^A) + (U-A)\cdot U - 2H^2 - \Sigma \cdot \Sigma, \\ 
            R_{0A} &= e_0(A_A) + \HH(U_A + A_A) - (A^B + U^B) \Sigma_{AB}, \\ 
            R_{A0} &= - e_A(H) + e_B \left(\Si AB \right) - 2A^B \Sigma_{AB}, \\ 
            R_{AB} &= - e_A(U_B) + e_0(\HH \delta_{AB} + \Sigma_{AB}) + (U-A)\cdot A \ \delta_{AB} \\ 
            &\quad - U_A (U_B + A_B) + e_C (A^C) \delta_{AB} + 2H(H \delta_{AB} + \Sigma_{AB}). \nonumber 
        \end{align}
    \end{subequations}
\end{lemma} 

Let $\tau = e^\phi$, then the $(2+1)$-dimensional modified Einstein-scalar-field system (\ref{ESForiginal_1})-(\ref{ESForiginal_2}) is equivalent to 
\begin{subequations}
    \begin{align}
        R_{ab}(g) &= T_{ab}(\tau) = \tau^{-1} \nabla_a \nabla_b \tau, \label{ESF01} \\ 
        \Box_g \tau &= 0. \label{ESF02} 
    \end{align}
\end{subequations}

\subsection{The choice of the background metric} 

As mentioned in the introduction, there are explicit solutions of (\ref{EVEoriginal}); for instance, the Kasner family $\cg$: 
\begin{equation} \label{Kasner}
    \cg = -(\mathrm dt)^2 + \sum_{\Lambda=1}^3 t^{2p_\Lambda} (\mathrm dx^\Lambda)^2, 
\end{equation}
where the constants $p_I$, called the \emph{Kasner exponents}, satisfy the \emph{Kasner relations}: 
\begin{equation} \label{Kasner_3dim}
    \sum_{\Lambda=1}^3 p_\Lambda = \sum_{\Lambda=1}^3 p_\Lambda^2 = 1. 
\end{equation}

We choose $\cg$ as our background spacetime. Then a simple algebraic check shows that there is exactly one of the exponents $p_\Lambda$'s negative; without loss of generality, we assume  
$$ p_1 < 0 < p_2 \, \& \, p_3. $$

\begin{remark}
    Recalling \defref{Symmetry}, it is necessary to take $p_3 > 0$. This choice ensures the compatibility with the Schwarzschild solution (see (\ref{Schwarzschild})), in which case the $x^3$ symmetry direction corresponds to an angular symmetry. 
    Moreover, it is also guaranteed that the singularity still occurs at the hypersurface $\Sigma_0$, and, as will be immediately seen below, that the following change of coordinates is well-defined. 
\end{remark}

Considering a new time function $t := t^{p_3}$, the background metric can be re-expressed as 
\begin{equation}
    \ch = - t^{r_0} \dd t^2 + \sum_{\Lambda = 1}^2 t^{r_\Lambda} (\mathrm d x^\Lambda)^2, \quad \ctau = t, 
\end{equation}
where the new Kasner exponents $r_\Lambda$'s satisfy 
\begin{equation}
    \begin{gathered} \label{Kasner_exponents}
        r_0 = \frac 2P - 4 \ge 0, \quad r_\Lambda = \frac {2p_\Lambda}{p_3}, \quad P = \frac {p_3}{1+p_3}, \\ 
        r_1 + r_2 = r_0, \quad r_1^2 + r_2^2 = \frac 2P r_0. 
    \end{gathered}
\end{equation}
Then it follows that 
    \begin{equation} \label{ValueofrC} 
        r_1 < 0 < r_2, \quad r_0 - r_1 > r_0 - r_2 > -2. 
    \end{equation}

Moreover, the background geometric quantities appearing in our decomposition (\ref{e0eA})-(\ref{eAeB}) are 
\begin{subequations}
    \begin{gather}
        \cnn = t^{r_0/2}, \quad \cee_A^\Lambda = t^{-r_A/2} \delta^\Lambda_A, \quad \cA_A = \cU_A = 0, \label{background_neAU}  \\ 
        \cHH = \frac {r_0}{4} t^{-r_0/2 - 1}, \quad \cSigma_{AB} = \left( \frac{r_{AB}}{2} - \frac {r_0}{4} \delta_{AB} \right) t^{-r_0/2 - 1},  \label{background_HSigma} 
    \end{gather}
\end{subequations}
where $r_{AB} = r_A$ if $A=B$ and $r_{AB} = 0$ otherwise. 

We remark that all subsequent calculations and the entire argument remain valid if working in the original time function, where many additional terms involving the parameter $p_3-1$ would then appear.

\subsection{The gauge freedom}

After imposing (\ref{e0eA})-(\ref{eAeB}), only one gauge freedom remains in the lapse function $\nn$. Aligned with \cite{beyer2025localized, isenberg2002asymptotic}, we fix this freedom by choosing the \emph{harmonic time} of spacetimes, $t=\tau$. Then (\ref{ESF01})-(\ref{ESF02}) become 
\begin{subequations}
    \begin{align}
        R_{ab}(g) &= T_{ab}(g) = t^{-1} \nabla_a \nabla_b t, \label{ESF1} \\ 
        \Box_g t &= 0. \label{ESF2} 
    \end{align}
\end{subequations}

Taking (\ref{e0eA}) into account, the components of $T_{ab}$ are 
\begin{subequations}
    \begin{gather}
        T_{00} = -2 t^{-1} \nn^{-1} \HH, \label{T00} \\ 
        T_{A0} = T_{0A} = - t^{-1} \nn^{-1} U_A, \label{TA0} \\ 
        T_{AB} = - t^{-1} \nn^{-1} \left( \HH \delta_{AB} + \Sigma_{AB} \right). \label{TAB} 
    \end{gather}
\end{subequations}

\begin{remark}
    In \cite{fournodavlos2023stable, an2025stabilitybigbangsingularity}, the authors fix the gauge freedom by using a constant mean curvature (CMC) foliation, 
    \begin{equation*}
        \tr k = - \frac 1t, 
    \end{equation*}
    where $k$ is the second fundamental form of the $t$-hypersurfaces $\Sigma_t$.
    This then yields an elliptic equation for $\nn$, which is convenient for energy estimates therein. 

    However, in our harmonic time gauge and the choice of time coordinates, we will instead recover similar properties for our solutions; see assertions (2) and (3) in our main theorem \thmref{Main} for the concrete expressions of $\tr k$ near singularities. 
\end{remark}

\subsection{The derivation of equations}

\begin{subequations}
    \begin{itemize}
        \item[a.] From \lemref{Ricci}, (\ref{ESF1}) and (\ref{T00})-(\ref{TAB}), we obtain 
        \begin{align}
            2 e_0(H) - e_A(U^A) &= 2 t^{-1} \nn^{-1} \HH + (U-A)\cdot U - 2H^2 - \Sigma \cdot \Sigma, \label{Hamiltonian_Constraint} \\ 
            e_0(A_A) &= - t^{-1} \nn^{-1} U_A - \HH(U_A + A_A) + (A^B + U^B) \Sigma_{AB}, \label{Momentum_Constraint1} \\ 
            e_A(\HH) - e_B \left(\Si AB \right) &= t^{-1} \nn^{-1} U_A - 2A^B \Sigma_{AB}, \label{Momentum_Constraint2} \\ 
            2 e_0(\HH) - e_A(U^A) + 2e_A (A^A) &= - 2\nn^{-1} t^{-1} \HH - 4H^2 + U \cdot U + 2 A \cdot A - U \cdot A, \label{Trace_RAB} \\ 
            e_0(\Sigma) - e \htimes U &= U \htimes (U+A) - 2H \Sigma - t^{-1} \nn^{-1} \Sigma, \label{Traceless_RAB} 
        \end{align}
        where (\ref{Trace_RAB}) and (\ref{Traceless_RAB}) correspond to the trace and the \emph{symmetric} traceless part of $R_{AB}$, respectively. The Hamiltonian and Momentum constraints are contained in (\ref{Hamiltonian_Constraint})-(\ref{Momentum_Constraint2}). Here and thereafter, for $1$-forms $A_I$ and $B_J$, we define their \emph{symmetric} traceless tensor product by 
        $$(A \htimes B)_{IJ} := \frac 12 (A_I B_J + A_J B_I) - \frac {\delta_{IJ}}{\tr \delta} \, \tr(A \otimes B); $$
        \item[b.] The relativistic $(2+1)$-dimensional orthonormal frame decompositions (\ref{e0eA})-(\ref{eAeB}) give immediately that 
        \begin{align}
            e_{[A} \left( e_{B]}^\Lambda \right) &= A_{[A} e_{B]}^\Lambda, \\ 
            e_0 \left( e_A^\Lambda \right) &= - \left( \HH \delta_A^B + \Si AB \right) e_B^\Lambda, \\ 
            e_A(\nn) &= \nn U_A; 
        \end{align}
        \item[c.] The Jacobi identity 
        $$\sum_{cyc} \  [[e_A, e_B], e_0] = 0,$$ 
        should also be reconcilable with (\ref{e0eA})-(\ref{eAeB}), which implies 
        \begin{align}
            e_{[A} (U_{B]}) &= A_{[A} U_{B]}; 
        \end{align}
        \item[d.] Finally, the gauge condition (\ref{ESF2}) gives the evolution equations for $\nn$ and $U_A$: 
        \begin{align}
            e_0 (\nn) &= 2 \nn \HH, \\ 
            e_0 (U_A) &= 2 e_A(\HH) + \HH U_A - U^B \Sigma_{AB}. 
        \end{align}
\end{itemize}
\end{subequations}

\begin{remark}[Key symmetrization] \label{KeySymmetrization} 
    The symmetrization step in (a) is essential and implicitly utilizes the Jacobi identity. Without symmetrization we will obtain the linear hyperbolic system for the propagation of constraints, that is, however, neither strongly hyperbolic nor Friedrichs-symmetrizable (see Section \ref{Section_Propagation_of_Constraints}), which introduces substantial difficulties. 
\end{remark}

\smallskip 

We now re-organize and summarize those equations as follows: 

\smallskip 

{\bf Constraint equations: }
\begin{subequations}
    \begin{align}
        e_{[A} \left( e_{B]}^\Lambda \right) &= A_{[A} e_{B]}^\Lambda, \label{constraint_e} \\ 
        e_{[A} (U_{B]}) &= A_{[A} U_{B]}, \label{constraint_U} \\ 
        e_A(\nn) &= \nn U_A, \label{constraint_nn} \\ 
        e_A(\HH) - e_B \left(\Si AB \right) &= - 2A^B \Sigma_{AB} + t^{-1} \nn^{-1} U_A, \label{constraint_HHandSigma} \\ 
        e_A(A^A) &= A \cdot A - H^2 + \frac 12 \Sigma \cdot \Sigma - 2 \nn^{-1} t^{-1} \HH; \label{constraint_A}
    \end{align}
\end{subequations}

\smallskip 

{\bf Evolution equations:} 
\begin{subequations}
    \begin{align}
        e_0(\nn) &= 2 \nn \HH, \label{evolution_nn} \\ 
        e_0 \left( e_A^\Lambda \right) &= - \left( \HH \delta_A^B + \Si AB \right) e_B^\Lambda, \label{evolution_e} \\ 
        e_0(A_A) + e_A(\HH) - e_B \left( \Si AB \right) &= - \HH (U_A + A_A) - (A^B - U^B) \Sigma_{AB}, \label{evolution_A} \\ 
        e_0 (U_A) - 2 e_A(\HH) &= \HH U_A - U^B \Sigma_{AB}, \label{evolution_U} \\ 
        2e_0(\HH) - e_A(U^A) + 2 e_A(A^A) &= - 4H^2 + U \cdot U + 2 A \cdot A - U \cdot A - 2 \nn^{-1} t^{-1} \HH,  \label{evolution_HH} \\ 
        e_0(\Sigma) - e \htimes U &= U \htimes (U+A) - 2H\Sigma - \nn^{-1} t^{-1} \Sigma. \label{evolution_Sigma} 
    \end{align}
\end{subequations}

\subsection{Symmetrization of evolution system}

Since the hyperbolic system (\ref{evolution_nn})-(\ref{evolution_Sigma}) is not symmetric right now, we modify it to obtain a symmetric quasi-linear hyperbolic system. 

The strategy is to use the Hamiltonian constraint (\ref{constraint_A}) and the Momentum constraints (\ref{constraint_HHandSigma}) as follows: 
\begin{itemize}
    \item[a.] multiply (\ref{constraint_HHandSigma}) by $\alpha$ and then add to (\ref{evolution_U}); 
    \item[b.] multiply (\ref{constraint_A}) by $\beta_{AB}$ and then add to (\ref{evolution_Sigma}), where $\beta_{AB}$ is some symmetric traceless tensor; 
    \item[c.] define new variables $B$ and $V$ to be the linear combinations of $A$ and $U$, namely $B = aA+bU$, $V = cA+dU$, with to-be-determined constants $a,b,c,d$. 
\end{itemize}

Considering these calculations, we shall require 
$$\alpha = 2, \quad \beta_{AB} = 0, \quad a+2b = -c, \quad \text{and} \quad a=c+2d,$$
to obtain the symmetric system, thus one convenient choice is 
$$\alpha = 2, \quad \beta = 0, \quad a=1, \quad b=0, \quad c=-1, \quad d=1.$$

Therefore, adding twice of (\ref{constraint_HHandSigma}) to (\ref{evolution_U}), we have the following symmetrized evolution equations in new variables $(\nn, e_A^\Lambda, A_A, V_A, \HH, \Sigma_{AB})$ with $V_A=U_A-A_A$: 

\smallskip 

{\bf Symmetrized evolution equations:} 
\begin{subequations}
    \begin{align}
        e_0(\nn) &= 2 \nn \HH, \label{evolution_nn'} \\ 
        e_0 \left( e_A^\Lambda \right) &= - \left( \HH \delta_A^B + \Si AB \right) e_B^\Lambda, \label{evolution_e'} \\ 
        e_0(A_A) + e_A(\HH) - e_B \left( \Si AB \right) &= - \HH (2A_A+ V_A) + V^B \Sigma_{AB}, \label{evolution_A'} \\ 
        e_0(V_A) - e_A(\HH) - e_B \left( \Si AB \right) &= \HH (A_A + 2V_A) - (5A^B + 2V^B) \Sigma_{AB} + 2 t^{-1} \nn^{-1} (A_A+V_A), \label{evolution_U'} \\ 
        2e_0(\HH) - e_A(V^A) + e_A(A^A) &= - 4\HH^2 + V \cdot V + 2 A \cdot A + V \cdot A - 2 \nn^{-1} t^{-1} \HH, \label{evolution_HH'} \\ 
        e_0(\Sigma) - e \htimes V + e \htimes A &= (A+V) \htimes (2A+V) - 2 \HH \Sigma - \nn^{-1} t^{-1} \Sigma, \label{evolution_Sigma'}  
    \end{align}
\end{subequations}

\smallskip 


\section{The Fuchsian system} \label{Section_Fuchsian_system}

In this section, we convert the evolution system into a Fuchsian form. Since the system has been symmetrized as (\ref{evolution_nn'})-(\ref{evolution_Sigma'}), we define $\WW$ to be the unknown vector field 
\begin{equation}
    \WW := \left( \nn, e_P^\Lambda, A_P, V_P, \HH, \Sigma_{PQ} \right)^T. 
\end{equation}
Then we arrive at the following symmetric quasi-linear hyperbolic system for $\WW$: 
\begin{equation} 
    \BB^0 \pa_t \WW + \nn e_C^\Lambda \BB^C \pa_\Lambda \WW = \GG(t,\WW), \label{HyperbolicSystemWW}
\end{equation}
where 
\begin{subequations}
    the coefficient matrices are 
    \begin{align}
        \BB^0 &:= \diag \left\{ 1, \delta^P_A, \delta^P_A, \delta^P_A, 2, \delta^P_A \delta^Q_B \right\}, \label{def_BB0} \\ 
        \BB^C &:= \begin{pmatrix}
            0 & 0 \\ 
            0 & \BB^C_*
        \end{pmatrix}, \label{def_BBC} \\ 
        \BB^C_* &:= \begin{pmatrix}
            0 & 0 & - \delta^C_A & \delta^{\langle P}_{A} \delta^{Q \rangle C} \\ 
            0 & 0 & \delta^C_A & \delta^{\langle P}_{A} \delta^{Q \rangle C} \\ 
            - \delta^{CP} & \delta^{CP} & 0 & 0 \\ 
            \delta^C_{\langle A} \delta^P_{B \rangle} & \delta^C_{\langle A} \delta^P_{B \rangle} & 0 & 0 
        \end{pmatrix}, \label{def_BBC*} 
    \end{align} 
    and the nonlinear terms are 
    \begin{align}
        \GG(t,\WW) := \left( G_1, G_2, G_3, G_4, G_5, G_6 \right)^T, 
    \end{align}
    with 
    \begin{align*}
        G_1 &:= 2 \nn^2 \HH, \\ 
        G_2 &:= - \left( \HH \delta_A^B + \Si AB \right) \nn e_B^\Lambda, \\ 
        G_3 &:= - \nn \HH (2A_A+ V_A) + \nn V^B \Sigma_{AB}, \\ 
        G_4 &:= \nn \HH (A_A + 2V_A) - (5A^B + 2V^B) \nn \Sigma_{AB} + 2 t^{-1} (A_A+V_A), \\ 
        G_5 &:= - 4\nn \HH^2 + \nn \left(V \cdot V + 2 A \cdot A + V \cdot A\right) - 2 t^{-1} \HH, \\ 
        G_6 &:= \nn (A+V) \htimes (2A+V) - 2\nn \HH \Sigma - t^{-1} \Sigma. 
    \end{align*}
\end{subequations}

Anticipating that the leading asymptotic behavior of unknown vector fields will be encapsulated by the background quantities (\ref{background_neAU}) and (\ref{background_HSigma}), we introduce the change of variables: 
\begin{subequations}
    \begin{gather}
        \tn := t^{\eps_0 - r_0/2} \nn, \label{change_tn} \\ 
        \te^\Lambda_I := t^{\eps_1 + r_I/2 - r_0/2} \nn e_I^\Lambda, \\ 
        \tA_A := t \nn A_A, \\ 
        \tU_A := t \nn U_A, \\ 
        \tH := t \nn \HH - \frac {r_0}{4}, \\ 
        \tSigma_{AB} := t \nn \Sigma_{AB} - \left( \frac{r_{AB}}{2} - \frac {r_0 \delta_{AB}}{4} \right), \label{change_tSigma}
    \end{gather}
\end{subequations}
where $\eps_0, \eps_1>0$ are sufficiently small parameters to be determined, and $r_{AB}$ as defined before.

The hyperbolic system (\ref{evolution_nn'})-(\ref{evolution_Sigma'}) then transforms into the following Fuchsian-form hyperbolic system: 
\begin{equation} 
    \tBB^0 \pa_t \tWW = \sum_C t^{r_0/2 - r_C/2 - \eps_1} \te_C^\Lambda \tBB^C \pa_\Lambda \tWW + \frac 1t \scrB \PP \tWW + \frac 1t \tFF, \label{FuchsianSystemtWW}
\end{equation}
where the unknown vector $\tWW$ is 
\begin{align}
    \tWW &:= \left( \tn, \te_P^\Lambda, \tA_P, \tV_P, \tH, \tSigma_{PQ} \right)^T; \label{def_tWW}
\end{align}
\begin{subequations}
    the coefficient matrices $\tBB^i$ are 
    \begin{align}
        \tBB^0 &:= \diag \left\{ 1, \delta^P_A, \delta^P_A, \delta^P_A, 2, \delta^P_A \delta^Q_B \right\}, \label{def_tBB0} \\ 
        \tBB^C &:= \begin{pmatrix}
            0 & 0 \\ 
            0 & \tBB^C_*
        \end{pmatrix}, \label{def_tBBC} \\ 
        \tBB^C_* &:= \begin{pmatrix}
            0 & 0 & - \delta^C_A & \delta^{\langle P}_{A} \delta^{Q \rangle C} \\ 
            0 & 0 & \delta^C_A & \delta^{\langle P}_{A} \delta^{Q \rangle C} \\ 
            - \delta^{CP} & \delta^{CP} & 0 & 0 \\ 
            \delta^C_{\langle A} \delta^P_{B \rangle} & \delta^C_{\langle A} \delta^P_{B \rangle} & 0 & 0 
        \end{pmatrix}; \label{def_tBBC*} 
        \end{align} 
    the constant matrices $\PP$ and $\scrB$ are 
        \begin{align}
        \PP &:= \diag \left\{ 1, \delta^P_A, \delta^P_A, \delta^P_A, 0, 0 \right\}; \label{def_PP} \\ 
        \scrB &:= \begin{pmatrix}
            \scrB_* & 0 \\ 
            0 & I_2 
        \end{pmatrix}, \label{def_scrB} \\ 
        \scrB_* &:= \begin{pmatrix}
            \eps_0 & & & \\ 
            & \eps_1 & & \\ 
            & & 1 + \frac {r_0}{2} - \frac {r_A}{2} & \\ 
            & & 2-3r_A+2r_0 & 3\left( 1 + \frac {r_0}{2} - \frac {r_A}{2} \right) 
        \end{pmatrix}; \label{def_scrB*} 
    \end{align} 
    and the nonlinear terms $\tFF$ are 
    \begin{align} 
        \tFF &:= \left( \tn \tH, (\tH + \tSigma)*\te^\Lambda, (\tH + \tSigma)*\tA, \right. \label{def_tFF}  \\ 
        &\quad \left. (\tH + \tSigma)*(\tA + \tV), \tA * (\tA + \tV), \tA * (\tA + \tV) \right)^T, \nonumber 
    \end{align}
    where $A*B$ denotes a (contracted) tensor field $A \otimes B$ with constant coefficients. 
\end{subequations}

\smallskip


\section{Propagation of the constraints} \label{Section_Propagation_of_Constraints}

In this section, we treat the propagation of the constraint equations, which is achieved by analyzing the associated linear hyperbolic system. To this aim, we define the constraints 
\begin{subequations}
    \begin{align}
        \EEE_{AB}^\Lambda &:= e_{[A} \left( e_{B]}^\Lambda \right) - A_{[A} e_{B]}^\Lambda, \label{def_constraintEEE} \\ 
        \UUU_{AB} &:= e_{[A} (U_{B]}) - A_{[A} U_{B]}, \label{def_constraintUUU} \\ 
        \NNN_A &:= \nn^{-1} e_A(\nn) - U_A, \label{def_constraintNNN} \\ 
        \MMM_A &:= e_A(\HH) - e_B \left(\Si AB \right) + 2A^B \Sigma_{AB} - t^{-1} \nn^{-1} U_A, \label{def_constraintMMM} \\ 
        \HHH &:= e_A(A^A) - A \cdot A + H^2 - \frac 12 \Sigma \cdot \Sigma + 2 \nn^{-1} t^{-1} \HH. \label{def_constraintHHH} 
    \end{align}
\end{subequations}
To avoid potential singularities involving the lapse function $\nn$ in the derivation of the following propagation equations, we include the factor $\nn^{-1}$ in the definition of the constraint $\NNN_A$ (\ref{def_constraintNNN}). 

\medskip

The following lemma is the main result of this section:

\begin{lemma} \label{ConstraintPropagation}
    Let $s \ge 5$. Consider the system (\ref{HyperbolicSystemWW}) with initial data 
    \begin{equation} \label{initialdata} 
        \left.\WW\right|_{t=t_0} = \WW_0 \in H^s(\Sigma_{t_0}), 
    \end{equation}
    then there exist a $t_1 \in (0,t_0)$ and a unique classical solution 
    $$\WW \in \bigcap_{j=0}^s C^j((t_1,t_0],H^{s-j}(\Sigma_t)).$$ 
    Moreover, there hold 
    \begin{itemize}
        \item[(1)] If the solution $\WW$ satisfies         $$\Norm{\WW}_{L^\infty((t_1,t_0]\times W^{1,\infty}(\Sigma_t))} < \infty,$$
        then $\WW$ can be continuously extended to a solution in $$\bigcap_{j=0}^s C^j((t_1',t_0],H^s(\Sigma_t))$$ 
        for some $t_1'<t_1$; 
        \item[(2)] Suppose the constraints $(\EEE, \UUU, \NNN, \MMM, \HHH) = 0$ vanish initially on $\Sigma_{t_0}$. Then 
        $$(\EEE, \UUU, \NNN, \MMM, \HHH) \equiv 0$$ 
        holds on the entire domain of existence of the solutions. 
    \end{itemize}
\end{lemma}

\subsection{Evolution equations for the constraints} 

After lengthy calculations, the following evolution equations hold for the constraints (\ref{def_constraintEEE})-(\ref{def_constraintHHH}): 
\begin{subequations}
    \begin{align}
        e_0 \left( \EEE_{AB}^\Lambda \right) &= - 2 \HH \EEE_{AB}^\Lambda + 2 \Si {[A}{C} \EEE_{B]C}^\Lambda, \label{propagation_E} \\ 
        e_0 \left( \MMM_A \right) &= e^B (\UUU_{AB}) - e_A \left( \HHH \right) - 2U_A \HHH - t^{-1} \nn^{-1} \left( \HH \delta_A^B + \Si AB \right) \NNN_B \label{propagation_M} \\ 
        &\quad + 2(U^B + A^B) \UUU_{AB} + \left( (t^{-1} \nn^{-1} - 3 \HH) \delta_A^B - \Si AB \right) \MMM_B, \nonumber \\ 
        e_0 \left( \UUU_{AB}  \right) &= - 2 e_{[A} \left( \MMM_{B]} \right) + 2 (A_{[A} - U_{[A}) \MMM_{B]}, \label{propagation_U} \\ 
        e_0 \left( \HHH \right) &= - e_A \left( \MMM^A \right) + (A^A - 2U^A) \MMM_A + t^{-1} \nn^{-1} U^A \NNN_A - ( 4 \HH + 2 \nn^{-1} t^{-1} ) \HHH, \label{propagation_H} \\ 
        e_0 \left( \NNN_A \right) &= - \left( \HH \delta_A^B + \Si AB \right) \NNN_B + 2 \MMM_A. \label{propagation_N} 
    \end{align}
\end{subequations}

From (\ref{propagation_E}), we immediately see that $\EEE$ decouples from the remaining constraints, and hence 
\begin{equation} \label{EEE} 
    \EEE_{AB}^\Lambda \equiv 0, 
\end{equation}
whenever it vanishes initially. The remaining constraint variables are coupled together and satisfy a linear hyperbolic system. Set 
$$\XXX := \left(\MMM_P, \UUU_{PQ}, \HHH, \NNN_P \right),$$
then $\XXX$ obeys 
\begin{equation} \label{XXXequation}
    \AAA^0 \pa_t \XXX + \nn \AAA^R e_R^\Lambda \pa_\Lambda \XXX = \BBB \XXX, 
\end{equation}
where the coefficient matrices are given by 
\begin{subequations}
    \begin{gather}
        \AAA^0 = \diag \left\{ \delta_A^P, \delta_A^P \delta_B^Q, 1, \delta_A^P \right\}, \\ 
        \AAA^R = \begin{pmatrix}
            0 & - \delta_A^P \delta^{QR} & \delta_A^R & 0 \\ 
            2 \delta^R_{[A} \delta^P_{B]} & 0 & 0 & 0 \\ 
            \delta^{PR} & 0 & 0 & 0 \\ 
            0 & 0 & 0 & 0 
        \end{pmatrix}, \\ 
        \BBB = \begin{pmatrix}
            t^{-1} \delta_A^P - 3 \nn \HH \delta_A^P - \nn \Si AP & 2 \nn (U^Q+A^Q) \delta_A^P & -2 \nn U_A & - t^{-1} \left( \HH \delta_A^P + \Si AP \right) \\ 
            2 \nn (A_{[A} - U_{[A})\delta^P_{B]} & 0 & 0 & 0 \\ 
            \nn (A^P - 2 U^P) & 0 & -4 \nn \HH - 2 t^{-1} & t^{-1} U^P \\ 
            2 \nn \delta_A^P & 0 & 0 & - \nn \left( \HH \delta_A^P + \Si AP \right) 
        \end{pmatrix}. 
    \end{gather}
\end{subequations}

\subsection{A strongly hyperbolic linear system}

We now establish the local existence and uniqueness for (\ref{HyperbolicSystemWW}) and prove the uniqueness of the trivial solution to (\ref{XXXequation}). This is mainly based on the strong hyperbolicity of the linear system (\ref{XXXequation}). 

\begin{lemma} \label{EigenvalueofAAA}
    Let $\AAA$ be the principal symbol of (\ref{XXXequation}) defined by 
    $$\AAA := \AAA^I \xi_I.$$ 
    Then for any $\xi \ne 0$ and any point $(x,t) \in M_{t_1,t_0}$, the matrix $\AAA$ is diagonalizable with spectrum consisting of the constant eigenvalues $0, \pm |\xi|$ and constant multiplicities $5,2$ respectively. Consequently, the linear hyperbolic system (\ref{XXXequation}) is constantly hyperbolic and thus admits a symbolic symmetrizer. 
\end{lemma}

\begin{proof}
    The principal symbol of (\ref{XXXequation}) is 
    \begin{equation}
        \AAA = \begin{pmatrix}
            0 & - \delta_A^P \xi^Q & \xi_A & 0 \\ 
            2 \xi_{[A} \delta^P_{B]} & 0 & 0 & 0 \\ 
            \xi^{P} & 0 & 0 & 0 \\ 
            0 & 0 & 0 & 0 
        \end{pmatrix}, 
    \end{equation}
    then its characteristic polynomial factors as 
    \begin{equation*}
        \det(\lambda \, \id - \AAA) = \lambda^5 (\lambda - |\xi|)^2 (\lambda + |\xi|)^2. 
    \end{equation*}
    We can check by a standard linear algebra argument that the eigenspaces associated with $\lambda_1 = 0$ and $\lambda_{2, \pm} = \pm |\xi|$ can be described by  
    \begin{align*}
        E_1 &= \left\{ \left( 0, \UUU_{AB}, \HHH, \NNN_A \right): |\xi|^2 \HHH = \frac 12 \xi^A \xi^B (\UUU_{AB} + \UUU_{BA}), \ \  |\xi|^2 \xi^B \UUU_{AB} = \frac 12 \xi_A \xi^C \xi^D (\UUU_{AB} + \UUU_{BA}) \right\}, \\ 
        E_{2,\pm} &= \left\{ \left( \MMM_A, \UUU_{AB}, \HHH, 0 \right): |\xi| \UUU_{AB} = \pm 2 \xi_{[A} \MMM_{B]}, \ \ |\xi| \HHH = \pm \xi^A \MMM_A \right\}. 
    \end{align*}
    Hence $\AAA$ is diagonalizable with stated spectrum, 
    which proves the lemma. 
\end{proof}

\begin{proofoflem}{\ref{ConstraintPropagation}}
    The statement (1) follows from the classical theory on symmetric first order quasi-linear hyperbolic equations, see e.g. Chapter 10 of \cite{benzoni2006multi}.

    Under the regularity hypothesis in (1), the coefficients in (\ref{XXXequation}) lie in $C^0((t_1,t_0],H^s(\Sigma_t))$. Hence the uniqueness result in Theorem 2.7 of \cite{benzoni2006multi} applies, which, combined with (\ref{EEE}), yields the statement (2). 
\end{proofoflem}

\smallskip


\section{The main theorem and its proof} 

In this section, we establish the main theorem and present its proof. The past global existence and asymptotic behavior of solutions are obtained from the Fuchsian theorem, see Section \ref{Section_Fuchsian_Theorem} and Appendix \ref{Appendix_Fuchsian}. The statement of the main theorem \thmref{Main} appears in Section \ref{Section_the_main_theorem}. 

\subsection{The Fuchsian theorem} \label{Section_Fuchsian_Theorem}

The following proposition states the main result of this subsection. 

\begin{proposition} \label{Prop_Fuchsian}
    Suppose $s \ge 5$ and $\tWW_0 \in H^s(\mathbb T^2)$. Consider the Fuchsian system (\ref{FuchsianSystemtWW}) together with initial data (\ref{initialdata}). 
    Then there exists $\delta_0>0$ such that if $$\left\|\tWW_0\right\|_{H^s(\mathbb T^2)}<\delta_0,$$ 
    then there is a unique classical solution 
    \begin{equation*}
        \tWW \in C^0 \big( (0,t_0], H^s(\mathbb T^2) \big) \cap L^\infty \big( (0,t_0], H^s(\mathbb T^2) \big) \cap C^1 \big( (0,t_0], H^{s-1}(\mathbb T^2) \big). 
    \end{equation*}
    Moreover, the limit $\lim_{t \rightarrow 0^-} \PP^\perp \tWW(t)$ exists in $H^{s-1}(\mathbb T^2)$, denoted by $\PP^\perp \tWW(0)$; and the solution $\tWW$ satisfies the estimates: 
    \begin{subequations}
        \begin{gather}
            \Norm{\tWW(t)}_{H^s(\TT^2)}^2 + \int_t^{t_0} \frac {1}{\tau} \Norm{\PP \tWW(\tau)}_{H^s(\TT^2)}^2 \dd \tau \lesssim \Norm{\tWW_0}_{H^s(\TT^2)}^2, \label{Prop_energyestimate} \\ 
            \Norm{\PP \tWW(t)}_{H^{s-1}(\TT^2)} + \Norm{\PP^\perp \tWW(t) - \PP^\perp \tWW(0)}_{H^{s-1}(\TT^2)} \lesssim t^\zeta, \label{Prop_asymbehavior}
        \end{gather}
    \end{subequations}
    for all $t \in (0,t_0]$ and some constant $\zeta>0$. 
\end{proposition}

\begin{proof}
    The proposition follows from the abstract Fuchsian result \thmref{FuchsianTheorem} in Appendix \ref{Appendix_Fuchsian}, once we check all five Fuchsian conditions there. We verify these conditions below. 

    First recall (\ref{Kasner_exponents}) and (\ref{ValueofrC}). Define
    $$\tilde r_C := 1 + \frac {r_0}{2} - \frac {r_C}{2} - \eps_1.$$ 
    We choose $\eps_1$ so small that 
    \begin{equation} \label{Choice_eps1} 
        0 < \eps_1 < 1 + \frac {r_0}{2} - \frac {r_C}{2} = \frac {1-p_C}{p_3}, 
    \end{equation}
    which guarantees that $\tilde r_C>0$. 

    \smallskip 
    
    \begin{itemize}
        \item[(1 \& 2)] By definitions (\ref{def_tBB0}), (\ref{def_PP}) and (\ref{def_scrB}), the matrices $\PP$, $\tBB^0$ and $\scrB$ are 
        \begin{gather*}
            \PP := \diag \left\{ 1, \delta^P_A, \delta^P_A, \delta^P_A, 0, 0 \right\}; \\ 
            \tBB^0 := \diag \left\{ 1, \delta^P_A, \delta^P_A, \delta^P_A, 2, \delta^P_A \delta^Q_B \right\}, \\ 
            \scrB := \begin{pmatrix}
                \scrB_* & 0 \\ 
            0 & I_2 
            \end{pmatrix}, \\ 
            \scrB_* := \begin{pmatrix}
            \eps_0 & & & \\ 
            & \eps_1 & & \\ 
            & & 1 + \frac {r_0}{2} - \frac {r_A}{2} & \\ 
            & & 2-3r_A+2r_0 & 3\left( 1 + \frac {r_0}{2} - \frac {r_A}{2} \right) 
            \end{pmatrix}. 
        \end{gather*}
        So they are all constant, and thus automatically time-independent and covariantly constant. The projector property $\PP^2 = \PP$ holds, and we also find out that $\tBB^0$ and $\scrB$ have no non-negative eigenvalues, due to (\ref{ValueofrC}) again. Moreover, 
        $$\PP \tBB^0 \PP^\perp = \PP^\perp \tBB^0 \PP = 0, \quad \text{and} \quad [\PP, \scrB] = 0$$
        hold as well. 
        Consequently, the decaying and converging components of $\tWW$ are identified as 
        \begin{subequations}
            \begin{gather}
                \PP \tWW = \left( \tn, \te_I^\Lambda, \tA_A, \tU_A, 0, 0 \right), \label{DecayingVariables} \\ 
                \PP^\perp \tWW = \left( 0, 0, 0, 0, \tH, \tSigma_{AB} \right); \label{ConvergingVariables} 
            \end{gather}
        \end{subequations} 
         
        \item[(3)] From (\ref{def_tFF}), the nonlinear term $\tFF$ has the form 
        \begin{equation*}
            \tFF := \left( \tn \tH, (\tH + \tSigma)*\te^\Lambda, (\tH + \tSigma)*\tA, (\tH + \tSigma)*(\tA + \tV), \tA * (\tA + \tV), \tA * (\tA + \tV) \right)^T. 
        \end{equation*}
        Using (\ref{DecayingVariables}) and (\ref{ConvergingVariables}), we immediately see that the algebraic structures required by the Fuchsian conditions hold: 
        \begin{equation}
            \PP \tFF = (\PP^\perp \tWW) * (\PP \tWW), \quad \PP^\perp \tFF = (\PP \tWW) * (\PP \tWW); 
        \end{equation}
        \item[(4)] Set 
        $$\tBB^\Lambda := \sum_{C} \, t^{\tilde r_C} \te_C^\Lambda \tBB^C,$$ 
        with $\tBB^C$ as in (\ref{def_BBC}). From (\ref{def_tBBC}) and (\ref{def_tBBC*}), we have that 
        \begin{gather}
            \PP \tBB \PP^\perp = |t|^{-1+p} O(\PP \tWW), \quad \PP^\perp \tBB \PP = |t|^{-1+p} O(\PP \tWW), \quad \PP^\perp \tBB \PP^\perp = 0, 
        \end{gather}
        where $p := \min_C \{ \tilde r_C \} - \eps_1 > 0$; 
        \item[(5)] Lastly, a straightforward comuptation shows  
        \begin{equation}
            \Div \left( t^{r_0/2-r_C-2-\eps_1} \te_C^\Lambda \tBB^C \right) = t^{-1+p} O(\pa \tWW),  
        \end{equation}
        with $p$ defined as above.
        Therefore, the divergence condition is fulfilled. 
    \end{itemize}
    This verifies the hypotheses of \thmref{FuchsianTheorem}, and the proposition follows. 
\end{proof}

\subsection{The main theorem} \label{Section_the_main_theorem}

We are now ready to state a precise version of our main result. 

\begin{theorem} \label{Main}
    Consider the $(3+1)$-dimensional Einstein vacuum equations 
    \begin{gather} \label{EinsteinVacuum}
        {^{(3+1)}R}_{ab}(g) = 0 \quad \text{on \ $M^{3+1} \simeq \mathbb R^+ \times \mathbb T^3$.} 
    \end{gather}
    Let $\cg$ be a Kasner background solution to (\ref{EinsteinVacuum}): 
    \begin{equation}
        \cg = -(\mathrm dt)^2 + \sum_{I=1}^3 t^{2p_I} (\mathrm dx^I)^2, 
    \end{equation}
    where the Kasner exponents obey (\ref{Kasner_3dim}) and $p_1 < 0 < p_2 \, \& \, p_3$. 
    Fix a positive integer $s \ge 5$ and positive constants $\eps_0$ and $\eps_1$ with $\eps_1$ satisfying (\ref{Choice_eps1}). 
    Let $(\Sigma_{t=t_0} = \mathbb T^3, \gini, \kini)$ be polarized $U(1)$-symmetric initial data (in the sense of \defref{Symmetry}) for (\ref{EinsteinVacuum}), which satisfy the constraint system (\ref{constraint_e})-(\ref{constraint_A}). Denote by $\tWW_0$ the induced initial data: 
    \begin{equation}
        \tWW_0 := \left( ^{0}\tn, \, ^{0}\te^\Lambda_P, \, ^{0}\tA_P, \, ^{0}\tV_P, \, ^{0}\tH, \, ^{0}\tSigma_{PQ} \right)^T, 
    \end{equation}
    as given by (\ref{e0eA}), (\ref{eAeB}), (\ref{background_neAU}), (\ref{background_HSigma}), and the scalings (\ref{change_tn})-(\ref{change_tSigma}). 

    Then there exists $\delta_0>0$ such that, for every $\delta \in (0,\delta_0)$, if the initial data satisfies 
    \begin{equation*}
        \Norm{\tWW_0}_{H^s(\mathbb T^2)} < \delta, 
    \end{equation*}
    then the symmetrized evolution system admits a unique classical solution $\tWW$ on the entire past region $M_{0,t_0}$ with 
    \begin{equation*}
        \tWW := \left( \tn, \te^\Lambda_P, \tA_P, \tV_P, \tH, \tSigma_{PQ} \right)^T \in \, \bigcap_{k=0}^s C^k((0,t_0], H^{s-k}(\Sigma_t)) \subset C^1(M_{0,t_0}), 
    \end{equation*}
    satisfying $\tWW|_{t=t_0} = \tWW_0$ initially and all constraints (\ref{constraint_e})-(\ref{constraint_A}) on $M_{0,t_0}$. 

    \smallskip 
    
    Furthermore, the solutions enjoys the following properties: 
    \begin{itemize}
        \item[(1)] {\em Uniform control and refined asymptotics.} The solution $\tWW$ has the uniform bound 
        \begin{equation*}
            \Norm{\tWW(t)}_{H^s(\TT^2)} \lesssim \delta, \quad \text{for all $0< t \le t_0$, } 
        \end{equation*} 
        where the coefficient depends on $\delta_0$. There exists $\zeta > 0$ and functions $\tn_*, \tH_*, \tSigma_{*AB} \in H^{s-1}(\TT^2) \subset C^0(\TT^2)$ satisfying 
        \begin{subequations}
            \begin{gather}
                0 < \inf_{x \in \TT^2} \tn_*(x) < \sup_{x \in \TT^2} \tn_*(x) < \infty, \label{tn_*} \\ \Norm{\tH_*}_{L^\infty(\TT^2)} \lesssim \delta, \quad 
                \tSigma_{*AB} = \tSigma_{*BA}, \quad \delta^{AB} \tSigma_{*AB} = 0, \label{tH_*tSigma_*} \\ 
                \max \left\{ \Norm{\tH_*}_{H^{s-1}(\TT^2)}, \Norm{\tSigma_*}_{H^{s-1}(\TT^2)} \right\} \lesssim \delta \label{max_tH_*tSigma_*} 
            \end{gather}
        \end{subequations}
        such that the following asymptotic expansions hold in $H^{s-1}(\mathbb T^2)$: 
        \begin{subequations}
            \begin{gather}
                \tn = t^{2\tilde H_* + \eps_0} \tn_* \left( 1+ O(t^\zeta) \right), \label{Asymbehavior_tn} \\ \te_I^\Lambda = O(t^\zeta), \quad \tA = O(t^\zeta), \quad \tU = O(t^\zeta), \label{Asymbehavior_teAU} \\ 
                \tH = \tH_* + O(t^\zeta), \quad \tSigma_{AB} = \tSigma_{*AB} + O(t^\zeta); \label{Asymbehavior_tHtSigma} 
            \end{gather}
        \end{subequations}
        
        \item[(2)] {\em Reduced $(2+1)$-dimensional solution.} The pair 
        $$h = - \nn^2 (\mathrm d t)^2 + \sum_{\Sigma, \Omega = 1}^2 h_{\Sigma \Omega} \mathrm d x^\Sigma \mathrm d x^\Omega, \quad \phi := \log t$$ 
        defines a classical solution of the $(2+1)$-dimensional modified Einstein-scalar field system (\ref{ESForiginal_1}), (\ref{ESForiginal_2}) on $M_{0,t_0}/\mathbb T^1$, where the components of the metric $h$ are 
        $$\nn = t^{r_0/2-\eps_0} \tn, \quad h^{\Sigma \Omega} = \delta^{IJ} e_I^\Sigma e_J^\Omega, \quad e_I^\Lambda = t^{-r_I/2+\eps_0-\eps_1} \tn^{-1} \te_I^\Lambda. $$ 
        The second fundamental form $^{(2+1)}k_{AB}$ determined by $h$ and the $t$-hypersurfaces satisfies 
        \begin{gather*}
            ^{(2+1)}k^A_A = t^{-1-r_0/2-2 \tH_*} \tn_*^{-1} \left( \frac {r_0}{4} + \tH_* + O(t^\zeta) \right) 
        \end{gather*}
        and thus blows up as $t \rightarrow 0^+$. In particular, $t=0$ is a crushing singularity; 
        
        \item[(3)] {\em Physical $(3+1)$-dimensional solution.} The metric 
        $$g = - \nn^2 (\mathrm d t)^2 + \sum_{\Sigma, \Omega = 1}^2 g_{\Sigma \Omega} \mathrm d x^\Sigma \mathrm d x^\Omega + t^2 (\mathrm d x^3)^2 $$ 
        defines a classical polarized $U(1)$-symmetric solution of the $(3+1)$-dimensional Einstein vacuum equations on $M_{0,t_0}$, with components given by 
        \begin{gather*}
            \nn = t^{r_0/2 - \eps_0} \tn, \quad g^{\Sigma \Omega} = \delta^{IJ} e_I^\Lambda e_J^\Omega, \quad e_I^\Lambda = t^{- r_I/2 + \eps_0 - \eps_1} \tn^{-1} \te_I^\Lambda. 
        \end{gather*}
        The solution $g$ is also asymptotically Kasner with Kasner exponents $r_0(x)$, $r_1(x)$, $r_2(x)$, $r_3 = 1$. The second fundamental form $^{(3+1)}k_{AB}$ determined by $g$ and the $t$-hypersurfaces satisfies 
        \begin{equation*}
            ^{(3+1)}k^A_A = t^{-1-2 \tH_*} \tn_*^{-1} \left( \frac {r_0}{4} + \tH_* + O(t^\zeta) \right) 
        \end{equation*}
        and thus blows up as $t \rightarrow 0^+$, so $t=0$ is a crushing singularity for the physical spacetime as well; 

        \item[(4)] {\em Inextendibility and curvature blow-up.} The spacetime $\left( M_{0,t_0}, g \right)$ is $C^2$-inextendible at $t=0$. By Hawking's singularity theorem, it is past timelike geodesically incomplete. The Kretschmann scalar $^{(3+1)}R^{abcd} \, {^{(3+1)}R}_{abcd}$ blows up according to 
        \begin{gather*}
            ^{(3+1)}R^{abcd} \, ^{(3+1)}R_{abcd} = 2 t^{-4-8\tH_*} \tn_*^{-4} \left( \mathfrak R + O(t^\zeta) \right), \\ 
            \mathfrak R := \sum_{i=0}^2 r_i^2 + 24 \tH_*^2 + 12 r_0 \tH_* + 4 \tSigma_* \cdot \tSigma_* + 4 r^{AB} \tSigma_{*AB}, 
        \end{gather*}
        and if $(r_0, r_1, r_2) \ne (0,0,0)$, or equivalently $p_3 \ne 1$, then we have that 
        $$\inf_{x \in \TT^3} \mathfrak R > 0;$$ 
        while the curvature invariant $^{(2+1)}R^{ab} \, ^{(2+1)}R_{ab}$ for the $(2+1)$-dimensional modified Einstein-scalar field equations satisfies a similar blow-up law: 
        \begin{gather*}
            ^{(2+1)}R^{ab} \, ^{(2+1)}R_{ab} = \frac 14 t^{-4-8\tH_*} \tn_*^{-4} \left( \mathfrak R + O(t^\zeta) \right). 
        \end{gather*}
        Finally, the perturbed solutions exhibit asymptotically velocity-terms dominated (AVTD) behavior: spatial derivative terms in (\ref{evolution_e})-(\ref{evolution_Sigma}) decay faster than the dominate temporal terms, and hence become negligible as $t \rightarrow 0^+$. 
    \end{itemize}
\end{theorem}

\begin{proofofthm}{\ref{Main}} \ \ 
    
    \smallskip

    \underline{Asymptotic behaviors.} 

    \smallskip 

    \propref{Prop_Fuchsian} yields existence, uniqueness and asymptotic decompositions for the projected components: 
    $$\PP \tWW = \left( \tn, \te_I^\Lambda, \tA_A, \tU_A - \tA_A, 0, 0 \right), \quad \PP^\perp \tWW = \left( 0, 0, 0, 0, \tH, \tSigma_{AB} \right).$$ 
    In particular, we obtain (\ref{tH_*tSigma_*}),(\ref{max_tH_*tSigma_*}), (\ref{Asymbehavior_teAU}), (\ref{Asymbehavior_tHtSigma}) and 
    $$\tn = O(t^\zeta).$$ 
    Moreover, the uniform bound for $\tWW$ holds 
    \begin{equation}
        \Norm{\tWW(t)}_{H^k(\TT^2)} \lesssim \delta, \quad \text{for any $0 < t \le t_0$. } \label{NormtWW}
    \end{equation}

    From the evolution equation (\ref{evolution_nn'}), We see that 
    \begin{equation*}
        \pa_t \log \tn = t^{-1} \left( 2\tH_* + \eps_0 + O(t^\zeta) \right), 
    \end{equation*}
    and integrating it in time yields the refined profile (\ref{tn_*}) and (\ref{Asymbehavior_tn}). 

    \smallskip 

    \underline{Crushing singularities.} 

    \smallskip 

    The second fundamental form for reduced metric $h$ is 
    \begin{align}
        {^{(2+1)}k}_{AB} &= h(\nabla_{e_A} e_0, e_B) = \gamma_{AB0} = \HH \delta_{AB} + \Sigma_{AB}, \label{(2+1)Dim_k}
    \end{align}
    and using (\ref{Asymbehavior_tn}), (\ref{Asymbehavior_tHtSigma}), we obtain 
    \begin{align*}
        \tr \, ^{(2+1)}k 
        &= 2 t^{-1-2\tH_*-r_0/2} \tn_*^{-1} \left( \frac {r_0}{4} + \tH_* + O(t^\zeta) \right). 
    \end{align*}
    Combining (\ref{NormtWW}) and Sobolev embedding, we see that $\tr k$ blows up as $t \rightarrow 0^+$, hence $t=0$ is a crushing singularity. 
    
    The same conclusion holds for the physical metric $g$ by the analogous computation. 

    \smallskip 

    \underline{Asymptotically pointwise Kasner behavior.} 

    \smallskip 
    
    It can be shown, from (\ref{Kasner_exponents}), (\ref{(2+1)Dim_k}), that 
    \begin{gather*}
        \lim_{t \rightarrow 0^+} \Norm{2 t \nn ^{(2+1)}k_{AB} - \mathrm k_{AB}}_{L^\infty} = 0, \quad \text{with $\mathrm k_{AB} := 2 \tH_* \delta_{AB} + 2 \tSigma_{*AB} + r_{AB}$}, \\ 
        \left( \mathrm k_A^A \right)^2 - \mathrm k_A^B \mathrm k_B^A + 4 \mathrm k_A^A = 12 \tH_*^2 + (4r_0+16) \tH_* - 4 \tSigma_* \cdot \tSigma_* - 4 r^{AB} \tSigma_{*AB} = O(t^{\zeta}) 
    \end{gather*}
    The conclusion then follows from Definition 1.1 in \cite{Todd2024}.  

    \smallskip 

    \underline{Blow-ups for the curvature invariants.} 

    \smallskip 

    In view of \lemref{Ricci}, a lengthy but straightforward calculation shows 
    \begin{gather*}
        ^{(3+1)}R^{abcd} \, {^{(3+1)}R}_{abcd} = 8 t^{-2} \nn^{-2} \left( 6 \HH^2 - 2 U \cdot U + \Sigma \cdot \Sigma \right), 
    \end{gather*}
    and substituting the expansions (\ref{Asymbehavior_tn})-(\ref{Asymbehavior_teAU}) with the rescalings (\ref{change_tn})-(\ref{change_tSigma}) yields the stated asymptotic form of the Kretschmann scalar. The $(2+1)$-dimensional $^{(2+1)}R^{ab} \, {^{(2+1)}R}_{ab}$ is handled similarly. Moreover, the positivity statement $\inf_{x \in \mathbb T^3} \mathfrak R > 0$ follows from $\sum_{i=0}^2 r_i^2 > 0$ after choosing $\delta_0$ sufficiently small and applying Sobolev embeddings. 
\end{proofofthm}

\smallskip


\begin{appendix}

\section{The Fuchsian method} \label{Appendix_Fuchsian}

The classical Fuchsian system in the literature, see e.g. \cite{kichenassamy2007fuchsian}, is presented as a semi-linear hyperbolic PDE system on $\mathbf R^n$ of the form: 
\begin{equation} \label{FuchsianSystem}
    t\pa_t \uu + A(t,x) \uu = \FF[\uu], 
\end{equation}
where $\uu = \uu(x,t)$ is a vector-valued unknown, the coefficient matrices and the vector-valued nonlinear terms $\FF[\uu] = \FF(t,x,\uu,\uu_x)$ are all analytic, and $\FF$ decays in $t$ at some positive rate. 

Without loss of generality, we assume the structural simplifications 
\begin{equation} \label{FuchsianCondition}
    A = A(x) \text{ only depends on $x$, and } \FF = t\ff(t,x,\uu,\uu_x). 
\end{equation}
Indeed, if $\FF$ decays like $t^\alpha$, we can rescale time by $t'=t^\alpha$ to recover the form (\ref{FuchsianCondition}). 
Likewise, if $A$ also depends on $t$, we can write $A = A_0(x) + tA_1(x,t,\uu)$, and absorb the $tA_1$ part into the right-hand side; 
and if $\uu$ possesses further decay, we may set $\uu'= t^\beta \uu$ to shift the spectrum of $A$ at the expense of changing the decay of the source. 

As shown in \cite{kichenassamy2007fuchsian}, the singular initial value problem for analytic data is locally well-posed: 
\begin{theorem} \label{Kiche_wellposedness}
    The system (\ref{FuchsianSystem}) subject to the conditions (\ref{FuchsianCondition}) is locally well-posed: there exists a unique analytic solution near $t=0$ which tends to zero as $t\rightarrow 0$. 
\end{theorem}

By relaxing analyticity, we arrive at a more general Fuchsian framework. A version suitable for our purposes (see Section \ref{Section_Fuchsian_system}), and following Beyer-Oliynyk-Olvera-Santamar{\'\i}a \cite{beyer2021fuchsian}, takes the form 
\begin{equation} \label{General_Fuchsiansystem}
    \A^0 \pa_t \uu + \A^I \nabla_I \uu = \frac {1}{t} \mathcal A \PP \uu + \frac 1t \FF[t, \uu]. 
\end{equation}
Here, the spatial derivative term $\A^I \nabla_I \uu$ is decoupled from the nonlinear source $\FF$. Moreover, since we are working on a manifold, all objects should be understood as sections and bundle maps over the vector bundles. But for brevity, we continue to use function notations below. 

\smallskip 

We state a streamlined Fuchsian theorem tailored to the setting of Section \ref{Section_Fuchsian_system}. 

\begin{theorem} \label{FuchsianTheorem}
    Let $M = \mathbf R \times \Sigma$ be an $(n+1)$-dimensional product manifold with closed spatial slice $\Sigma$. 
    Fix an integer $s > n/2+3$ and a constant $\sigma > 0$. 
    Consider the PDE system (\ref{General_Fuchsiansystem}) with initial data $$\uu(t_0, \cdot) = \uu_0 \in H^s(\Sigma). $$
    Then there exists $\delta>0$ such that if 
    \begin{equation*}
        \|\uu_0\|_{H^s(\Sigma)} \le \delta,
    \end{equation*} 
    then the system (\ref{General_Fuchsiansystem}) admits a unique classical solution
    $$\uu \in C_b^0([T_0,0),H^s(\Sigma)) \cap C^1([T_0,0), H^{s-1}(\Sigma)),$$ 
    such that the limit $\PP^{\perp} \uu(0) := \lim_{t \rightarrow 0} \PP^\perp \uu(t)$ exists in $H^{s-1}(\Sigma)$. Moreover, the solution $\uu$ satisfies the following \emph{energy estimates}: 
    \begin{equation}
        \Norm{\uu(t)}_{H^s(\Sigma)} - \int_{T_0}^t \frac {1}{\tau} \Norm{\PP \uu(\tau)}_{H^s(\Sigma)}^2 \dd \tau \le C(\delta, \delta^{-1}) \Norm{\uu_0}_{H^s(\Sigma)}^2, 
    \end{equation}
    and \emph{decay estimates}: 
    \begin{subequations}
        \begin{gather}
            \Norm{\PP \uu(t)}_{H^{s-1}} \lesssim |t|^{\zeta - \sigma}, \\ 
            \Norm{\PP^\perp \uu(t) - \PP^\perp \uu(0)}_{H^{s-1}} \lesssim |t|^{\zeta - \sigma}, 
        \end{gather}
    \end{subequations} 
    where $\zeta := \kappa - \gamma_1 s (s-1) \mathtt A/2 > 0$ (the involving constants are given below). 
    These conclusions hold provided that the following \underline{Fuchsian conditions} are satisfied: 

\begin{subequations}
    \begin{itemize}
        \item[(1)] The matrix $\PP$ is a time-independent, covariantly constant and symmetric projection: 
        \begin{equation} \label{estimatePP}
            \PP^2 = \PP, \quad \PP^T = \PP, \quad \pa_t \PP = 0, \quad \nabla \PP = 0, 
        \end{equation}
        thus we can define $\PP^\perp := \text{id} - \PP$, satisfying the same above conditions and $\PP \PP^\perp = \PP^\perp \PP = 0$; 
        \item[(2)] The coefficient matrices $\A^0$ and $\mathcal A$ obey 
        \begin{gather}
            (\A^0)^T = \A^0, \quad \pa_t \A^0 = 0, \quad \nabla \A^0 = 0, \label{conditionA} \\ 
            \quad \pa_t \mathcal A = 0, \quad \nabla \mathcal A = 0, \label{conditionAcal} \\ 
            \PP \A^0 \PP^\perp = \PP^\perp \A^0 \PP = 0, \quad [\PP, \mathcal A] = 0, \label{conditionPA}
        \end{gather}
        and there exist positive constants $\kappa, \gamma_1, \gamma_2$ with 
        \begin{equation}
            \frac {1}{\gamma_1} \id \le \A^0 \le \frac {1}{\kappa} \mathcal A \le \gamma_2 \id; \label{A0_and_AAA}
        \end{equation}
        \item[(3)] The nonlinear term $\FF$ satisfies the bounds 
        \begin{equation}
            \PP \FF = O\left(\frac {\lambda_1}{R} \uu \otimes \PP \uu\right), \quad \PP^\perp \FF = O\left( \frac{\lambda_2}{R} \PP \uu \otimes \PP \uu\right), \label{estimateFF}
        \end{equation}
        for some constants $\lambda_1, \lambda_2 > 0$; 
        \item[(4)] The spatial flux map $\A := (\A^I)_{1 \le I \le n}$ is symmetric in the sense that  
        \begin{equation}
            [\sigma \A]^T = \sigma \A 
        \end{equation}
        for any spatial $1$-form $\sigma$, and satisfies the estimates 
        \begin{gather}
            \PP \A \PP = |t|^{-1} O(\uu), \quad \PP^\perp \A \PP^\perp = |t|^{-1} O(\PP \uu \otimes \PP \uu), \\ 
            \PP \A \PP^\perp = |t|^{-1} O(\PP \uu), \quad \PP^\perp \A \PP = |t|^{-1} O(\PP \uu); \label{estimateA}
        \end{gather} 
        \item[(5)] The divergence of $\A$ $\Div \A := \nabla_I \left( \A^I(t,x,\uu) \right)$ satisfies 
        \begin{gather}
            \Div \A = |t|^{-1+p} O(\theta), \label{estimateDivA}
        \end{gather}
        for some constants $0 < p \le 1$ and $\theta > 0$. 
    \end{itemize}    
\end{subequations}
    Finally, all $O$-symbols above, say $f = O(g)$, indicate the presence of coefficient functions $h$ such that $f = hg$ and 
    \begin{equation*}
        \Norm{\nabla^k_x h}_{L^\infty} < 1, \text{\quad for all \ $k \ge 0$}; 
    \end{equation*}
    moreover, all parameters appearing above satisfy 
    \begin{gather}
        2\kappa > \max \left\{ 2 \gamma_1(\lambda_1+\lambda_2), s(s+1) \mathtt A \right\}, \quad \mathtt A := \sup_{t_0 \le t < 0} \left( |t| \, \Norm{\nabla \PP \A \PP}_{L^\infty} \right) < \infty. \label{assumptions_kappa}
    \end{gather}
\end{theorem}

\smallskip 

\begin{proofofthm}{\ref{FuchsianTheorem}}
    The proof follows the strategy of Theorem~3.8 in \cite{beyer2021fuchsian}. We give a condensed presentation adapted to our simplified hypotheses. 
    
    We organize the argument in the following five steps: (i) derive $L^2$ and $H^s$ estimates for $\uu$, (ii) prove global existence on $[t_0,0)$, (iii) show convergence of $\PP^\perp \uu(t,\cdot)$ as $t \rightarrow 0^-$, (iv) obtain decay estimates for $\PP \uu(t, \cdot)$, and (v) establish decay for $\PP^\perp \uu(t, \cdot) - \PP^\perp \uu (0, \cdot)$. 
    Moreover, due to (\ref{A0_and_AAA}), we will use the fact that $\A^0$ is uniformly equivalent to the identity, so that 
    $$\angbracket{\vv, \A^0 \vv} \approx \Norm{\vv}_{L^2}$$ 
    for any function $\vv$. Here and henceforth, the angle bracket $\angbracket{\cdot, \cdot}$ denotes the usual inner product of the Hilbert space $L^2(\Sigma)$. 
    
    \smallskip 
    
    \underline{$L^2$ and $H^s$ energy estimates.} 

    \smallskip 

    Pair (\ref{General_Fuchsiansystem}) with $\uu$ and integrate to obtain 
    \begin{equation*}
        \pa_t \angbracket{\uu, \A^0 \uu} = \angbracket{\uu, \Div \A \uu} + \frac 2t \angbracket {\uu, \AAA \PP \uu} + \frac 2t \angbracket{\uu, \FF}. 
    \end{equation*}
    Using (\ref{A0_and_AAA}), (\ref{estimateFF}) and (\ref{estimateDivA}), and the covariant invariance of $\A^0$, $\AAA$ and $\PP$, we have 
    \begin{align*}
        \frac {1}{\gamma_1} \pa_t \Norm{\uu}_{L^2}^2 &\le \norm{\angbracket{\uu, \Div \A \uu}} - \frac {2}{|t|} \angbracket{\uu, \AAA \PP \uu} + \frac {2}{|t|} \norm{\angbracket{\uu, \FF}} \\ 
        &\le \theta |t|^{-1+p} \Norm{\uu}_{L^2}^2 - \frac {2 \kappa}{\gamma_1 |t|} \Norm{\PP \uu}_{L^2}^2 + \frac {2(\lambda_1 + \lambda_2)}{|t|} \Norm{\PP \uu}_{L^2}^2. 
    \end{align*}
    Rearranging this inequality gives the $L^2$ estimate of $\uu$:  
    \begin{equation} \label{L2_estimate}
        \pa_t \Norm{\uu}_{L^2}^2 \le \theta \gamma_1 |t|^{-1+p} \Norm{\uu}_{L^2}^2 - |t|^{-1} \rho_0 \Norm{\PP \uu}_{L^2}^2, 
    \end{equation}
    where, by the assumption (\ref{assumptions_kappa}), the constant $\rho_0$ is positive 
    $$\rho_0 := 2\kappa - 2\gamma_1(\lambda_1+\lambda_2) > 0.$$

    Differentiating (\ref{General_Fuchsiansystem}) $k$-times in space, for $1 \le k \le s$, we obtain that 
    \begin{equation*} 
        \A^0 \pa_t \nabla^k \uu + \A^I \nabla_I \nabla^k \uu = \frac 1t \AAA \nabla^k \PP \uu + \frac 1t \nabla^k \FF[t,\uu] - [\nabla^k, \A^I] \nabla_I \uu - \A^I [\nabla^k, \nabla_I] \uu, 
    \end{equation*}
    which, paired with $\nabla^k \uu$, yields that 
    \begin{equation} \label{sderivative}
        \pa_t \angbracket{\nabla^k \uu, \A^0 \nabla^k \uu} = \angbracket{\nabla^k \uu, \Div \A \nabla^k \uu} + \frac 2t \angbracket{\nabla^k \uu, \AAA \nabla^k \PP \uu} + 2 \angbracket{\nabla^k \uu, \GG_k}, 
    \end{equation}
    where $\GG_k$ collects 
    $$\GG_k := \frac 1t \nabla^k \FF[t,\uu] - [\nabla^k, \A^I] \nabla_I \uu - \A^I [\nabla^k, \nabla_I] \uu.$$

    The differentiated source term $\nabla^k \FF$ can be estimated by 
    \begin{align*}
        \norm{\angbracket{\nabla^k \uu, \nabla^k \FF[t, \uu]}} &\le \norm{\angbracket{\nabla^k \PP^\perp \uu, \nabla^k \PP^\perp \FF[t,\uu]}} + \norm{\angbracket{\nabla^k \PP \uu, \nabla^k \PP \FF[t,\uu]}} \\ 
        &\le \frac {\lambda_1 + \lambda_2}{R}  \Norm{\uu}_{H^s} \Norm{\PP \uu}_{H^s}^2, 
    \end{align*}
    with structure (\ref{estimateFF}) applied in the second inequality. 
    The commutator term $[\nabla^k, \nabla_I] \uu$ can be decomposed by: 
    \begin{align*}
        \angbracket{\nabla^k \uu, \A^I [\nabla^k, \nabla_I] \uu} &= \angbracket{\nabla^k \PP \uu, \PP \A^I \PP [\nabla^k, \nabla_I] \PP \uu} + \angbracket{\nabla^k \PP \uu, \PP \A^I \PP^\perp [\nabla^k, \nabla_I] \PP^\perp \uu} \\ 
        &\quad + \angbracket{\nabla^k \PP^\perp \uu, \PP^\perp \A^I \PP [\nabla^k, \nabla_I] \PP \uu} + \angbracket{\nabla^k \PP^\perp \uu, \PP^\perp \A^I \PP^\perp [\nabla^k, \nabla_I] \PP^\perp \uu}. 
    \end{align*}
    Applying (\ref{estimateA}) and since $[\nabla^k, \nabla_I]$ is a differential operator of order $k-1$, together with our choice $s>n/2+3$, we obtain the bound 
    \begin{align*}
        \norm{\angbracket{\nabla^k \uu, \A^I [\nabla^k, \nabla_I] \uu}} &\le \Norm{\nabla^k \PP \uu}_{L^2} \left( |t|^{-1} \Norm{\PP \uu}_{K^{k-1}} \Norm{\uu}_{H^s} \right) + \Norm{\nabla^k \PP \uu}_{L^2} \left( |t|^{-1} \Norm{\uu}_{H^{k-1}} \Norm{\PP \uu}_{H^s} \right) \\ 
        &\quad + \Norm{\nabla^k \uu}_{L^2} \left( |t|^{-1} \Norm{\PP \uu}_{H^{k-1}} \Norm{\PP \uu}_{H^s} \right) + \Norm{\nabla^k \uu}_{L^2} \left( |t|^{-1} \Norm{\PP \uu}_{H^s}^2 \Norm{\uu}_{H^{k-1}} \right) \\ 
        &\le |t|^{-1} C(\Norm{\uu}_{H^s}) \Norm{\PP \uu}_{H^s}^2, 
    \end{align*}
    where, here and thereafter, $C (\eta)$ denotes a coefficient polynomially depending on $\eta > 0$ \footnote{and, possibly on $s$ and the fixed parameters in Fuchsian conditions.}, and tending to $0$ with positive rate as $\eta \rightarrow 0^+$. 
    The other commutator term $[\nabla^k, \A^I] \nabla_I \uu$ in $\GG_k$ can be similarly decomposed as 
    \begin{align*}
        \angbracket{\nabla^k \uu, [\nabla^k, \A^I] \nabla_I \uu} &= \angbracket{\nabla^k \uu, \sum_{j=1}^k \nabla^j \A^I * \nabla^{k-j} \nabla_I \uu} \\ 
        &= \angbracket{\nabla^k \PP \uu, \sum_{j=1}^k \nabla^j (\PP \A^I \PP) * \nabla^{k-j} \nabla_I \PP \uu} \\ 
        &\quad + \angbracket{\nabla^k \PP \uu, \sum_{j=1}^k \nabla^j (\PP \A^I \PP^\perp) * \nabla^{k-j} \nabla_I \PP^\perp \uu} \\ 
        &\quad + \angbracket{\nabla^k \PP^\perp \uu, \sum_{j=1}^k \nabla^j (\PP^\perp \A^I \PP) * \nabla^{k-j} \nabla_I \PP \uu} \\ 
        &\quad + \angbracket{\nabla^k \PP^\perp \uu, \sum_{j=1}^k \nabla^j (\PP^\perp \A^I \PP^\perp) * \nabla^{k-j} \nabla_I \PP^\perp \uu}. 
    \end{align*}
    With similar arguments and (\ref{estimateA}) again, we find that 
    \begin{align*}
        \norm{\angbracket{\nabla^k \uu, [\nabla^k, \A^I] \nabla_I \uu}} &\le 
        k \Norm{\nabla \PP \A \PP}_{L^\infty} \Norm{\PP \uu}_{H^s}^2 + |t|^{-1} C(\Norm{\uu}_{H^s}) \Norm{\PP \uu}_{H^s}^2, 
    \end{align*}
    where the first term and second term come from $l=1$ and $l\ge2$ respectively. 

    Hence the estimates for the $\GG_k$-terms are given by 
    \begin{align*}
        \norm{\angbracket{\nabla^k \uu, \GG_k}} &\le \norm{\angbracket{\nabla^k \uu, \A^I [\nabla^k, \nabla_I] \uu}} + \angbracket{\nabla^k \uu, [\nabla^k, \A^I] \nabla_I \uu} + |t|^{-1} \norm{\angbracket{\nabla^k \uu, \nabla^k \FF[t, \uu]}} \\ 
        &\le |t|^{-1} k \mathtt A \Norm{\PP \uu}_{H^s}^2 + |t|^{-1} C \left( \Norm{\uu}_{H^s} \right) \Norm{\PP \uu}_{H^s}^2. 
    \end{align*}
    Combining this with (\ref{sderivative}), we conclude the $L^2$ estimate for $\nabla^k \uu$: 
    \begin{align*}
        \pa_t \Norm{\nabla^k \uu}_{L^2}^2 &\le - |t|^{-1} 2\kappa \Norm{\nabla^k \PP \uu}_{L^2}^2 + \gamma_1 \theta |t|^{-1+p} \Norm{\nabla^k \uu}_{L^2}^2 \\ 
        &\quad + |t|^{-1} 2k \gamma_1 \mathtt A \Norm{\PP \uu}_{H^s}^2 + |t|^{-1} C \left(\Norm{\uu}_{H^s}\right) \Norm{\PP \uu}_{H^s}^2, 
    \end{align*}
    with $1 \le k \le s$. Summing over $k$, with help of (\ref{L2_estimate}), we thus arrive at the $H^s$ estimate for $\uu$: 
    \begin{equation} \label{Hs_estimate}
        \pa_t \Norm{\uu}_{H^s}^2 \le - |t|^{-1} \rho_s \Norm{\PP \uu}_{H^s}^2 + \theta \gamma_1 |t|^{-1+p} \Norm{\uu}_{H^s}^2 + |t|^{-1} 2 \gamma_1 (\lambda_1+\lambda_2) \Norm{\PP \uu}_{L^2}^2, 
    \end{equation}
    with $\rho_s$ defined by 
    \begin{gather*}
        \rho_s := 2\kappa - \gamma_1 s(s+1) \mathtt A - C \left( \Norm{\uu}_{H^s} \right).
    \end{gather*}

    \smallskip 

    \underline{Global existence on $[t_0,0)$.} 

    \smallskip 

    Adding $\rho_0^{-1} 2 \gamma_1(\lambda_1+\lambda_2)$ times (\ref{L2_estimate}) to (\ref{Hs_estimate}), we can remove the appearing $\Norm{\uu}_{L^2}$ term, which yields  
    \begin{equation}
        \pa_t \left( \Norm{\uu}_{H^s}^2 + \rho_0^{-1} 2 \gamma_1 (\lambda_1+\lambda_2) \Norm{\uu}_{L^2}^2 \right) \le - |t|^{-1} \rho_s \Norm{\PP \uu}_{H^s}^2 + C_0 |t|^{-1+p} \Norm{\uu}_{H^s}^2, 
    \end{equation}
    where, here and thereafter, $C_0$ denotes a positive constant depending only on $s$ and the fixed parameters in the Fuchsian conditions. 
    
    Define the energy function 
    $$E_s(t) := \Norm{\uu}_{H^s}^2 + 2 \rho_0^{-1} \gamma_1 (\lambda_1+\lambda_2) \Norm{\uu}_{L^2}^2 - \int_{t_0}^t \frac {\rho_s}{\tau} \Norm{\PP \uu(\tau)}_{H^s}^2 \dd \tau,$$ 
    then it satisfies 
    $$\pa_t E_s(t) \le C_0 |t|^{-1+p} E_s(t).$$
    Gronwall's inequality then produces that 
    \begin{equation} \label{enerygyinequality}
        E_s(t) \le e^{C_0(|t_0|^p - |t|^p)} E_s(t_0). 
    \end{equation}

    By choosing $\delta_0$ sufficiently small, one ensures that 
    $$C \left( \Norm{\uu}_{H^s} \right) < \kappa - \frac 12 \gamma_1 s(s+1) \mathtt A \quad \text{for all \ $0 < \Norm{\uu}_{H^s} \le 2\delta_0$},$$
    then our assumption (\ref{assumptions_kappa}) gives uniformly that 
    $$\rho_0 > 0, \quad \rho_s > 0.$$
    Since initially we have 
    \begin{equation*}
        \Norm{\uu(t_0)}_{H^s} \le \delta, 
    \end{equation*} 
    refining $\delta < \delta_0$, and defining 
    $$t_{\delta_0} := \sup \left\{ t \in (t_0, 0]: \Norm{\uu(t)}_{H^s} \le 2 \delta_0 \right\}$$ 
    to be the first time so that $\Norm{\uu(t_{\delta_0})}_{H^s} = 2\delta_0$, the bootstrap argument provides that 
    \begin{align*}
        \Norm{\uu(t)}_{H^s}^2 &\le e^{C_0 |t_0|^p} \left[ \Norm{\uu(t_0)}_{H^s}^2 + 2 \rho_0^{-1} \gamma_1 (\lambda_1+\lambda_2) \Norm{\uu(t_0)}_{L^2}^2 \right] \\ 
        &\le e^{C_0 |t_0|^p} [1 + 2 \rho_0^{-1} \gamma_1 (\lambda_1+\lambda_2)] \delta^2.  
    \end{align*}
    for all $t \in [t_0, t_{\delta_0})$, by (\ref{enerygyinequality}). Now let 
    $$\delta < e^{- C_0 |t_0|^p/2} [1+\rho_0^{-1}\gamma_1(\lambda_1+\lambda_2)]^{-1/2} \delta_0$$ 
    be small enough, 
    it then follows that $\Norm{\uu(t)}_{H^s} < \delta_0$, contradicting the choice of $t_{\delta_0}$; and hence $t_{\delta_0} = 0$, and the solution exists globally on $[t_0, 0) \times \Sigma$ with uniform bounds 
    \begin{equation*} 
        \Norm{\uu(t)}_{H^s} \le e^{C_0 (1+|t_0|^p)} \delta. 
    \end{equation*}
    Furthermore, we have the following energy estimate for $\uu$: 
    \begin{equation} \label{estimate_uu}
        \Norm{\uu(t)}_{H^s}^2 - \int_{t_0}^t \frac {1}{\tau} \Norm{\PP \uu(\tau)}_{H^s}^2 \dd \tau \le e^{C_0(1+|t_0|^p-|t|^p)} \Norm{\uu(t_0)}_{H^s}^2. 
    \end{equation}

    \smallskip 

    \underline{Convergence of $\PP^\perp \uu$.} 

    \smallskip 

    Projecting (\ref{General_Fuchsiansystem}) onto $\PP^\perp$ direction gives 
    \begin{equation} \label{PPperpProjection}
        \pa_t \PP^\perp \uu = - \PP^\perp (\A^0)^{-1} \A^I \nabla_I \uu + \frac 1t \, \PP^\perp (\A^0)^{-1} \left( \AAA \PP \uu + \FF[t,\uu] \right), 
    \end{equation}
    where the right-hand side can be decomposed further, by using (\ref{estimatePP}), (\ref{conditionPA}), as 
    \begin{equation*} 
        \pa_t \PP^\perp \uu = - \PP^\perp (\A^0)^{-1} \PP^\perp \A^I \PP \nabla_I \PP \uu - \PP^\perp (\A^0)^{-1} \PP^\perp \A^I \PP^\perp \nabla_I \PP^\perp \uu + \frac 1t \, \PP^\perp (\A^0)^{-1} \PP^\perp \FF[t,\uu].  
    \end{equation*} 
    Integrating in time over $[t_1,t_2] \subset [t_0, 0)$ and applying the $H^{s-1}$ norm, we estimate that 
    \begin{equation} \label{PPperp_uu}
        \begin{aligned} 
            \Norm{\PP^\perp \uu(t_2) - \PP^\perp \uu(t_1)}_{H^{s-1}} &\le \gamma_1 \int_{t_1}^{t_2} \Norm{\PP^\perp \A \PP \nabla \PP \uu}_{H^{s-1}} + \Norm{\PP^\perp \A \PP^\perp \nabla \PP^\perp \uu}_{H^{s-1}} + \Norm{\PP^\perp \FF}_{H^{s-1}} \dd t \\ 
            &\le \gamma_1 \int_{t_1}^{t_2} |t|^{-1} \left( 2 + \Norm{\uu(t)}_{H^{s-1}} \right) \Norm{\PP \uu(t)}_{H^s} \Norm{\PP \uu(t)}_{H^{s-1}} \dd t, 
        \end{aligned}
    \end{equation}
    where (\ref{estimateFF}), (\ref{estimateA}) have been used in the second inequality. 

    The energy bound (\ref{estimate_uu}) implies 
    \begin{align*}
        \Norm{\PP^\perp \uu(t_2) - \PP^\perp \uu(t_1)}_{H^{s-1}} &\le C_0 \left( 1 + e^{|t_0|^p} \delta \right) \int_{t_1}^{t_2} - \frac 1t \Norm{\PP \uu}_{H^s}^2 \dd t, 
    \end{align*}
    and hence tends to $0$ as $t_1, t_2 \rightarrow 0^-$. Therefore, $\PP^\perp \uu$ converges in $H^{s-1}$ with limit denoted by 
    $$\PP^\perp \uu(0) := \lim_{t \rightarrow 0^-} \PP^\perp \uu(t) \in H^{s-1}(\Sigma).$$

    \smallskip 

    \underline{Decay estimates of $\PP \uu$.} 

    \smallskip 

    Projecting (\ref{General_Fuchsiansystem}) onto $\PP$ direction yields, in view of (\ref{estimatePP}), (\ref{conditionPA}), that 
    \begin{equation} \label{PP_uu}
        \PP \A^0 \PP \pa_t \PP \uu + \PP \A^I \PP \nabla_I \PP \uu + \PP \A^I \PP^\perp \nabla_I \PP^\perp \uu = \frac 1t \AAA \PP \uu + \frac 1t \PP \FF[t,\uu], 
    \end{equation}
    Pairing this with $\PP \uu$ gives 
    \begin{align*}
        \frac 12 \pa_t \angbracket{\PP \uu, \A^0 \PP \uu} &= \frac 12 \angbracket{\PP \uu, \Div \A \PP \uu} + \frac 1t \angbracket{\PP \uu, \AAA \PP \uu} + \angbracket{\PP \uu, \frac 1t \PP \FF - \PP \A^I \PP^\perp \nabla_I \PP^\perp \uu}.
    \end{align*}
    The last term can be bounded with usage of (\ref{estimateFF}) and (\ref{estimateA}): 
    \begin{align*}
        \norm{\angbracket{\PP \uu, \frac 1t \PP \FF - \PP \A^I \PP^\perp \nabla_I \PP^\perp \uu}} &\le |t|^{-1} \left( \lambda_1 + \Norm{\uu}_{H^s} \right) \Norm{\PP \uu}_{L^2}^2; 
    \end{align*}
    we then obtain the $L^2$ decay estimate of $\PP \uu$: 
    \begin{equation} \label{L2estimate_PPuu}
        \pa_t \Norm{\PP \uu}_{L^2}^2 \le |t|^{-1+p} \gamma_1 \theta \Norm{\PP \uu}_{L^2}^2 - |t|^{-1} \tilde 2 \rho_0 \Norm{\PP \uu}_{L^2}^2, 
    \end{equation}
    where, as before, by (\ref{assumptions_kappa}) and choosing $\delta_0$ sufficiently small, one ensures that 
    $$\tilde \rho_0 := \kappa - C \left( \Norm{\uu}_{H^s} \right) > 0.$$ 
    Gronwall's inequality then produces that 
    \begin{equation} 
        \Norm{\PP \uu(t)}_{L^2} \le C_0 e^{|t_0|^p/2} \norm{\frac {t}{t_0}}^{\tilde \rho_0} \delta. 
    \end{equation}

    Differentiating (\ref{PP_uu}) in space again, we have that 
    \begin{equation}
        \begin{aligned}
            \PP \A^0 \PP \pa_t \nabla^k \PP \uu + \PP \A^I \PP \nabla_I \nabla^k \PP \uu &= \frac 1t \AAA \nabla^k \PP \uu + \tilde \GG_k, 
        \end{aligned}
    \end{equation}
    where $\tilde \GG_k$ is defined by 
    \begin{equation*}
        \begin{aligned}
            \tilde \GG_k :&= \frac 1t \nabla^k \PP \FF - \PP \A^I \PP^\perp \nabla_I \nabla^k \PP^\perp \uu \\ 
            &\quad - [\nabla^k, \PP \A^I \PP] \nabla_I \PP \uu - [\nabla^k, \PP \A^I \PP^\perp] \nabla_I \PP^\perp \uu \\ 
            &\quad - \PP \A^I \PP [\nabla^k, \nabla_I] \PP \uu - \PP \A^I \PP^\perp [\nabla^k, \nabla_I] \PP^\perp \uu. 
        \end{aligned}
    \end{equation*}
    Pairing with $\nabla^k \PP \uu$ then gives 
    \begin{equation}
        \frac 12 \pa_t \angbracket{\nabla^k \PP \uu, \PP \A^0 \PP \nabla^k \PP \uu} = \frac 12 \angbracket{\nabla^k \PP \uu, \Div A \nabla^k \PP \uu} + \frac 1t \angbracket{\nabla^k \PP \uu, \AAA \nabla^k \PP \uu} + \angbracket{\nabla^k \PP \uu, \tilde \GG_k}, 
    \end{equation}
    Repeating the commutator estimates again \footnote{Notice that there is a loss of one derivative, due to $\nabla_I \nabla^k \PP^\perp \uu$ terms.}, it can be shown that 
    \begin{align*}
        \pa_t \Norm{\nabla^k \PP \uu}_{L^2} 
        &\le - |t|^{-1} 2 \kappa \Norm{\nabla^k \PP \uu}_{L^2}^2 + \gamma_1 \theta |t|^{-1+p} \Norm{\nabla^k \PP \uu}_{L^2}^2 \\ 
        &\quad + |t|^{-1} C \left( \Norm{\uu}_{H^s} \right) \Norm{\PP \uu}_{H^{s-1}}^2 + |t|^{-1} 2k \gamma_1 \mathtt A \Norm{\PP \uu}_{H^{s-1}}^2. 
    \end{align*}
    Taking summation over $0 \le k \le s-1$ and in view of (\ref{L2estimate_PPuu}), we obtain that 
    \begin{align*}
        \pa_t \Norm{\PP \uu}_{H^{s-1}}^2 
        &\le |t|^{-1+p} \gamma_1 \theta \Norm{\PP \uu}_{H^{s-1}}^2 - |t|^{-1} 2 \tilde \rho_s \Norm{\PP \uu}_{H^{s-1}}^2 + |t|^{-1} C \left( \Norm{\uu}_{H^s} \right) \Norm{\PP \uu}_{L^2}^2, 
    \end{align*}
    with the positive exponent and positive gap 
    $$\tilde \rho_s := \kappa - \frac 12 \gamma_1 s(s-1) \mathtt A - C \left( \Norm{\uu}_{H^s} \right) > 0, \quad \tilde \rho_0 - \tilde \rho_s > 0,$$ 
    by choosing $\delta_0$ sufficiently small. 
    By Gronwall's inequality again, together with (\ref{estimate_uu}) and (\ref{L2estimate_PPuu}), we derive the $H^{s-1}$ decay of $\PP \uu$: 
    \begin{equation} \label{HSestimate_PPuu}
        \Norm{\PP \uu}_{H^{s-1}} \le C(\delta) e^{C_0|t_0|^p} \norm{\frac {t}{t_0}}^{\tilde \rho_s}. 
    \end{equation}

    \smallskip 

    \underline{Decay estimate of $\left( \left. \PP^\perp \uu - \PP^\perp \uu \right|_{t=0} \right)$.} 

    \smallskip 

    Finally, recalling (\ref{PPperp_uu}) and letting $t_1 = t$ and $t_2 \rightarrow 0^-$ give that  
    \begin{align*}
        \Norm{\PP^\perp \uu(t) - \PP^\perp \uu(0)}_{H^{s-1}} 
        &\le \gamma_1 \int_{t}^{0} |t|^{-1} \left( 2 + \Norm{\uu(t)}_{H^{s-1}} \right) \Norm{\PP \uu(\tau)}_{H^s} \Norm{\PP \uu(\tau)}_{H^{s-1}} \dd \tau \\ 
        &\le C_0 \left( 1 + C(\delta) \right) \int_t^0 - \frac 1 \tau \Norm{\PP \uu(\tau)}_{H^s} \Norm{\PP \uu(\tau)}_{H^{s-1}} \dd \tau. 
    \end{align*}
    Inserting estimates in (\ref{HSestimate_PPuu}) and by Holder's inequality, we arrive at 
    \begin{equation}
        \Norm{\PP^\perp \uu(t) - \PP^\perp \uu(0)}_{H^{s-1}} \le C(\delta) e^{C_0|t_0|^p} \norm{\frac {t}{t_0}}^{\tilde \rho_s}. 
    \end{equation}
Taking $\zeta := \kappa - \gamma_1 s(s-1) \mathtt A/2$ to be the principal part of $\tilde \rho_s$ and since $\delta$ can be chosen sufficiently small with bound $\sigma$, we retrieve the converging estimate for $\PP^\perp\uu$. We complete the proof. 
\end{proofofthm}

\begin{remark}
    By the time change $-t'= (-t)^q$, we also obtain an analogous version of the theorem with an explicit parameter $0 < q \le 1$. 
\end{remark}

\end{appendix}

\printbibliography

\nocite{*}

\end{document}